\documentclass[11pt]{article}

\usepackage{amsfonts,algorithmic,algorithm}
\usepackage{graphicx,subfig}
\usepackage{amsmath,amsthm}
\usepackage{wrapfig}

\usepackage{epsfig}
\usepackage{amsmath}
\usepackage{amssymb}
\usepackage{psfrag}
\usepackage[absolute]{textpos}

\usepackage{hyperref}
\usepackage{amsfonts}
\usepackage{graphicx,subfig}
\usepackage{amsmath,amsthm}
\usepackage{wrapfig}
\usepackage{amssymb}

\usepackage{epsfig}
\usepackage{amsmath}
\usepackage{amssymb}
\usepackage{psfrag}
\usepackage[absolute]{textpos}
\usepackage{setspace}

\newtheorem{thm}{Theorem}[section]
\newtheorem{lem}{Lemma} [section]
\newtheorem{definition}{Definition}

\newcommand{\A}{{\bf A}}

\def\x{{\bf x}}
\def\y{{\bf y}}
\def\w{{\bf w}}

\def\z{{\bf z}}
\def\e{{\bf e}}

\newcommand{\beq}{\begin{equation}}
\newcommand{\eeq}{\end{equation}}
\newcommand{\bea}{\begin{eqnarray}}
\newcommand{\eea}{\end{eqnarray}}

\newcommand{\stexp}{\mbox{$\mathbb{E}$}}    
\newcommand{\Prob}{\ensuremath{\mathbb{P}}}

\long\def\symbolfootnote[#1]#2{\begingroup%
\def\thefootnote{\fnsymbol{footnote}}\footnote[#1]{#2}\endgroup}

\begin{document}

\title{Improving the Thresholds of Sparse Recovery: An Analysis of a Two-Step Reweighted Basis Pursuit Algorithm}
\author{ M. Amin Khajehnejad$^\dag$, Weiyu Xu$^*$, A. Salman Avestimehr$^*$ and Babak Hassibi$^\dag$ \\\text{}\\$^\dag$California Institute of Technology, Pasadena CA 91125 \\ $^*$Cornell University, Ithaca NY 14853
\thanks{The results of this paper were presented in part at the International Symposium on Information Theory, ISIT 2010}
\thanks{This work was supported in part by the National Science Foundation under grants CCF-0729203, CNS-0932428 and CCF-1018927, by the Office of Naval Research under the MURI grant N00014-08-1-0747, and by Caltech's Lee Center for Advanced Networking.}
}
 \maketitle
 \vspace{1cm}

\begin{abstract}
\boldmath
It is well known that $\ell_1$ minimization can be used to recover sufficiently sparse unknown signals from compressed linear measurements. In fact, exact thresholds on the sparsity, as a function of the ratio between
the system dimensions,
so that with high probability almost all sparse signals can be
recovered from i.i.d. Gaussian measurements, have been computed and are
referred to as ``weak thresholds'' \cite{D}. In this paper, we
introduce a reweighted $\ell_1$ recovery algorithm composed of two
steps: a standard $\ell_1$ minimization step to identify a set of
entries where the
signal is likely to reside, and a weighted $\ell_1$ minimization step where entries outside this set are penalized. For signals where the non-sparse component entries are independent and identically drawn from certain classes of distributions, (including most well known continuous distributions), we prove a
\emph{strict} improvement in the weak recovery threshold. Our analysis suggests that the level of improvement in the weak threshold depends on the behavior of the distribution at the origin. Numerical simulations
verify the distribution dependence of the threshold improvement very well, and suggest that in the case of i.i.d. Gaussian nonzero entries, the improvement can be quite impressive---over 20\%  in
the example we consider.
\end{abstract}


%

\section{Introduction}
\label{sec:Intro}
Compressed sensing addresses the problem of recovering sparse signals
from under-determined systems of linear equations  \cite{Rice}. In
particular, if $\x$ is an $n\times1$ real vector which is known to have at
most $k$ nonzero elements where $k<n$, and $\A$ is an $m\times n $
measurement matrix with $k<m<n$, then for appropriate values of $k$,
$m$ and $n$, it is possible to efficiently recover $\x$ from the set of linear projections $\y=\A\x$ \cite{D
  CS,DT,CT,Baraniuk Manifolds}. The most well recognized such
algorithm is $\ell_1$ minimization which can be formulated  as
follows:
\beq
\label{eq:l1 min}
\min_{\A\z=\A\x}\|\z\|_1.
\eeq
The first result that established the fundamental thresholds of signal
recovery using $\ell_1$ minimization is due to Donoho and Tanner
\cite{D,DT}, where it is shown
that if the measurement matrix is i.i.d. Gaussian, for a given ratio
of $\delta = \frac{m}{n}$, $\ell_1$ minimization can successfully
recover {\em every} $k$-sparse signal, provided that $\mu = \frac{k}{n}$ is
smaller than a certain threshold. This statement is true
asymptotically as $n\rightarrow \infty$ and with high
probability. This threshold guarantees the recovery of {\em all}
sufficiently sparse signals and is therefore referred to as a \emph{strong}
threshold. It therefore does not depend on the actual distribution of
the nonzero entries of the sparse signal and as such is a universal
result. However, at this point, it is not known whether there exist other polynomial-time algorithms with strong thresholds superior to those of $\ell_1$ minimization.

Another notion introduced and computed in \cite{D,DT} is that of a
{\em weak} threshold where signal recovery is guaranteed for {\em
  almost all} support sets and {\em almost all} sign patterns of the
sparse signal, with high probability as $n\rightarrow\infty$. The weak
threshold is the one that can be observed in simulations of $\ell_1$
minimization and allows for signal recovery beyond the strong threshold. The weak threshold of $\ell_1$ minimization is also universal from the vantage point of signal distribution; The amplitudes of the nonzero entries of a sparse signal does not affect its recoverability by solving (\ref{eq:l1 min}). In other words, if a sparse signal with a support set $S$ and a particular sign pattern is recoverable using $\ell_1$ minimization, so is every other signal with the same support and sign pattern. It is worth noting that the weak thresholds of $\ell_1$ minimization can be generalized to a broader class of random measurement matrices, including those with null spaces that are random orthant symmetric and generic subspaces (e.g., matrices with i.i.d. Bernoulli or uniform (-1,1) entries, etc.)~\cite{Donoho_face_counting}. Finally, similar to the strong thresholds, it is not known whether there exist other polynomial-time algorithms with superior weak thresholds than $\ell_1$ minimization.

\noindent \textbf{Our Contributions.} In this paper we prove that a certain \emph{two-step
  reweighted $\ell_1$} algorithm indeed has higher weak recovery
guarantees than ordinary $\ell_1$ minimization for particular classes of sparse signals, including sparse
Gaussian signals.  We had previously introduced this algorithm in
\cite{Khajehnejad_Allerton}, and had proven that for a very
restricted class of \emph{polynomially decaying} sparse signals it
outperforms standard $\ell_1$ minimization. In this paper however, we
  extend this result to a much wider and more reasonable class of
sparse signals. The key to our result is the fact that for these
  classes of signals, $\ell_1$ minimization has an \emph{approximate support
  recovery} property which can be exploited in reweighted
  $\ell_1$ algorithm, to obtain a provably superior weak threshold.
In particular, we consider Gaussian sparse
signals, namely sparse signals in which the nonzero entries are i.i.d.
  Gaussian. Our analysis of Gaussian sparse signals relies on
  concentration bounds on the partial sum of their order
  statistics. Furthermore, we show that
for continuous distributions with sufficiently fast decaying tails and
nonzero value at the origin, similar improvements for the weak threshold can be postulated. More generally, we show that as long as the nonzero entries of the sparse signal are independently drawn from a continuous distribution $f(\cdot)$ that has a nonzero finite order derivative at the origin, the weak recovery threshold of our proposed two step reweighted $\ell_1$ algorithm is strictly larger than that of $\ell_1$ minimization. Although not specifically derived, our analysis suggests that the improvement rate is a function of the smallest integer $r$ for which  $f^{(r)}(0)\neq 0$; The smaller such $r$ is, the larger the improvement is. We perform numerical simulations using various distributions which authenticate this assertion.

It is worth noting that different variations of reweighted $\ell_1$
algorithms have been recently introduced in the literature and, have
shown experimental improvement over ordinary $\ell_1$ minimization
\cite{Needell,CWB07}.  In \cite{Needell} approximately sparse signals
have been considered, where perfect recovery is often not achieved. The question is therefore not that of an explicit recovery threshold extension. Instead, it has been shown that the reconstruction error can be reduced using an iterative scheme. In \cite{CWB07}, a similar algorithm is suggested and is empirically shown to outperform $\ell_1$
minimization for exactly sparse signals with certain continuous
distributions. In particular, it was empirically witnessed that the proposed algorithm does not improve the signal recovery for sparse vectors with constant amplitude nonzero entries (\emph{i.e.} a nonzero entry is either 1 or -1).  Unfortunately, \cite{CWB07} provides no theoretical analysis or performance guarantees for the success or failure of the method. The particular reweighted $\ell_1$ minimization algorithm that we propose and analyze is of significantly
less computational complexity than the earlier ones (it only solves
two linear programs). Furthermore, experimental results confirm that
it exhibits much better performance than previous reweighted methods.
Finally, while we do rigorously establish a strict
{\em improvement} in the weak threshold, we currently do not have
tight bounds on the new weak threshold and simulation results are far
better than the bounds we can provide at this time.

The organization of this paper is as follows. In Section \ref{sec:def}, we introduce the basic definitions used throughout the paper. In Section \ref{sec:model}, the signal model is described, the notions of strong and weak recovery thresholds are quantified and the main problem is stated, namely to find a polynomial time recovery algorithm with better thresholds than $\ell_1$ minimization for sparse signal recovery. In Section \ref{sec:Algorithm} a two step reweighted linear programming algorithm is described and is claimed to be superior in performance to the regular $\ell_1$ minimization algorithm for sparse vectors with Gaussian distributions (Theorem \ref{thm: final thm}). Sections \ref{sec:robustness} and \ref{sec:perfect recovery} are dedicated to the detailed proof of this claim, through  separate analysis of different stages of the algorithm. In Section \ref{sec:generalization}, these results are generalized to a much broader class of sparsity models beyond  Gaussians. The technical discussions of this paper predict that the performance of the proposed algorithm strongly depends on the distribution of the nonzero entries of the random sparse signal model. The paper ends in Section \ref{sec:simulation} with some numerical evaluations of the proposed algorithm and the verification of the distribution dependent behavior of the reweighted algorithm.
\hfill 
\section{Basic Definitions}
\label{sec:def}
A sparse signal with exactly $k$ nonzero entries is called
$k$-sparse. For a vector $\x$, $\|\x\|_1$ denotes the $\ell_1$
norm. The support (set) of $\x$,  denoted by $supp(\x)$, is the index
set of its nonzero coordinates. For a vector $\x$ that is not exactly
$k$-sparse, we define the $k$-support of $\x$ to be the index set of
the largest $k$ entries of $\x$ in amplitude, and denote it by
$supp_k(\x)$. For a subset $K$ of the entries of $\x$, $\x_K$ means
the vector formed by those entries of $\x$ indexed in $K$. Finally,
$\max|\x|$ and $\min|\x|$ mean the absolute value of the maximum and
minimum entry of $\x$ in magnitude, respectively.

\section{Signal Model and Problem Description}
\label{sec:model}
We consider sparse random signals with i.i.d. nonzero
coefficients drawn from a given continuous distribution (in particular Gaussian). In other words we assume that the unknown sparse signal is an
$n\times 1$ vector $\x$ with exactly $k$ nonzero entries, where each
nonzero entry is independently derived from a distribution $f(\cdot)$ (\textit{e.g.},
standard normal distribution $\mathcal{N}(0,1)$). The measurement matrix $\A$ is an $m\times n$
matrix with i.i.d. Gaussian entries with an aspect ratio $\delta
= \frac{m}{n}$. The theory of compressed sensing guarantees that if
$\mu=\frac{k}{n}$ is smaller than a certain threshold, then for almost all measurement matrices $\A$ every
$k$-sparse signal can be recovered using $\ell_1$ minimization. The
relationship between $\delta$ and the maximum threshold of $\mu$ for
which such a guarantee exists is called the \emph{strong sparsity
threshold}, and is denoted by $\mu_{S}(\delta)$.  A more practical
performance guarantee is the so-called \emph{weak sparsity threshold},
denoted by $\mu_{W}(\delta)$, which has the following
interpretation: For a fixed value of $\delta = \frac{m}{n}$ and an
i.i.d. Gaussian matrix $\A$ of size $m\times n$,  a random $k$-sparse
vector $\x$ of size $n\times 1$ with a randomly chosen support set and
a random sign pattern can be recovered from $\A\x$ using $\ell_1$
minimization with high probability, if
$\frac{k}{n}<\mu_{W}(\delta)$. In addition, other forms of recovery thresholds can be defined using different constraints and requirements. For example, when the reconstruction of signals with all support sets and almost all sign patterns is considered, the resulting thresholds are called \emph{sectional}. These thresholds were discussed in \cite{DT} for i.i.d. Gaussian matrices. Furthermore, strong and weak thresholds can also be defined and evaluated for the reconstruction of nonnegative signals (see \emph{e.g.} \cite{Donoho positive}), or for alternative classes of matrix ensembles. For example, strong thresholds for $\ell_1$ minimization over expander-graph-based measurement matrices were derived in \cite{Indyk}, and in \cite{Khajehnejad_Expanders} for nonegative vectors in addition to weak threshold forms.

In this paper, we consider sparse signals that fall outside the recoverability regime of $\ell_1$ minimization. In other words, we assume that the support size of $\x$, namely $k$, is slightly
larger than the weak threshold of $\ell_1$ minimization. In other
words,  $k = (1+\epsilon_0)\mu_{W}(\delta)n$ for some
$\epsilon_0>0$. This means that if we use $\ell_1$ minimization, a
randomly chosen $\mu_{W}(\delta)n$-sparse signal will be recovered
perfectly with very high probability, whereas a randomly selected
$k$-sparse signal will not. We would like to show that for a strictly
positive $\epsilon_0$, the two-step reweighted $\ell_1$ algorithm of
Section \ref{sec:Algorithm} can indeed recover a randomly selected
$k$-sparse signal with high probability, implying that the proposed method has a
superior weak threshold.

\section{Two-Step Weighted $\ell_1$ Algorithm}
\label{sec:Algorithm}
We propose the following method outlined in Algorithm \ref{alg:modmain}, consisting of two linear programming steps: a standard $\ell_1$ minimization and a weighted one. The input to
the algorithm is the vector $\y=\A\x$, where $\x$ is the unknown $k$-sparse
signal with $k=(1+\epsilon_0)\mu_W(\delta)n$, and the output is an
approximation $\x^*$ to the unknown vector $\x$. We assume that the sparsity $k$ (or an upper bound on it) is known. However, the algorithm assumes no knowledge of the distribution of the nonzero entries of the unknown signal. Also $\omega>1$ is a predetermined weight.

\begin{figure}[t]
\centering
  \includegraphics[width= 0.4\textwidth]{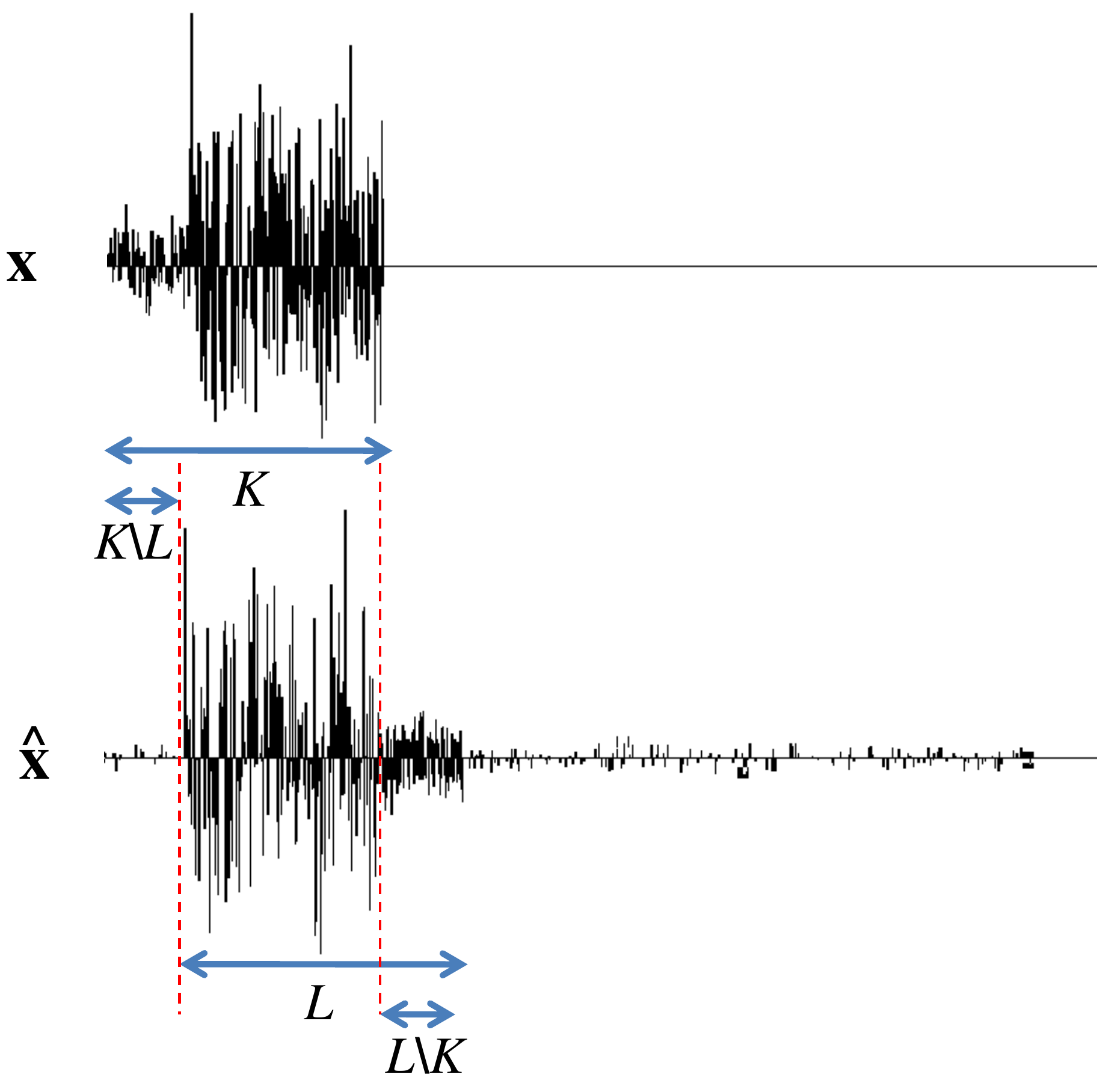}
  \caption{ \scriptsize{A pictorial example of a sparse signal  and its $\ell_1$ minimization approximation.}}
  \label{fig:sigmal}
\end{figure}

\begin{algorithm}[t]\caption{\small{Two Step Reweighted $\ell_1$ minimization.}}
\begin{algorithmic}[1]
\STATE \textbf{Input: } Measurement matrix $\A^{m\times n}$, measurement vector $\y^{m\times 1}$, integer $k<n$, predetermined real valued weight $\omega >1$.\\
\STATE \textbf{Output: } Sparse vector $\x$ with $\A\x = \y$.
\STATE Solve the $\ell_1$ minimization problem:
\begin{equation}
\hat{\x} = \arg{ \min{ \|\z\|_1}}~~\text{subject to}~~ \A\z
= \y.
\end{equation}
\STATE Obtain an approximation for the support set of $\x$:
find the index set $L \subset \{1,2, ..., n\}$ which corresponds to
the largest $k$ elements of
$\hat{\x}$ in magnitude.
\STATE Solve the following weighted $\ell_1$ minimization problem and declare the solution as output:
 \beq
\x^* = \arg{\min\|\z_L\|_1+\omega\|\z_{\overline{L}}\|_1}~~\text{subject to}~~ \A\z
= \y.
\label{eq:weighted l_1}
\eeq
\end{algorithmic}
\label{alg:modmain}
\end{algorithm}

The intuition behind the algorithm is as follows. In the first step, a standard $\ell_1$ minimization is performed. If the sparsity of the signal is beyond the weak threshold $\mu_W(\delta)n$, then $\ell_1$ minimization is most probably not capable of recovering the signal. However, we use
the output of the $\ell_1$ minimization to identify an index set, $L$,
which we ``hope'' contains most of the nonzero entries of $\x$ (see Figure \ref{fig:sigmal}). We
finally perform a weighted $\ell_1$ minimization by penalizing those
entries of $\x$ that are not in $L$ (ostensibly because they have a
lower chance of being nonzero). Consequently, Algorithm \ref{alg:modmain} is capable of recovering less sparse signals, or equivalently has a higher weak threshold than that of $\ell_1$ minimization. This intuition is formalized in the following theorem.

\begin{thm}[Weak threshold of Algorithm \ref{alg:modmain}]
Let $\A$ be an $m\times n$ i.i.d. Gaussian matrix with
$\frac{m}{n}=\delta$. There exist $\epsilon_0>0$ and $\omega>0$ so that Algorithm
\ref{alg:modmain} perfectly recovers a random
$(1+\epsilon_0)\mu_{W}(\delta)n$-sparse vector with i.i.d. Gaussian
entries with high probability as $n$ grows to infinity.
\label{thm: final thm}
\end{thm}

The interpretation of the above theorem is that for sparse signals whose nonzero entries follow a Gaussian distribution, Algorithm
\ref{alg:modmain} has a recovery threshold beyond that of standard
$\ell_1$ minimization. The proof is provided in the next sections as follows. In Section
\ref{sec:robustness}, we prove that there is a large overlap between
the index set $L$, found in step 2 of the algorithm, and the support set
of the unknown signal $\x$ (denoted by $K$)---see Theorem \ref{thm:l_1
  support recovery} and Figure \ref{fig:sigmal}. Then in Section
\ref{sec:perfect recovery}, we show that the large overlap between $K$
and $L$ can result in perfect recovery of $\x$, beyond the standard
weak threshold, when a weighted $\ell_1$ minimization is used in step
3. The formal proof of Theorem \ref{thm: final thm} appears in Section \ref{sec:perfect recovery}.

\section{Approximate Support Recovery, Steps 1 and 2 of the Algorithm}
\label{sec:robustness}

In this section, we carefully study the first two steps of Algorithm
\ref{alg:modmain}. The unknown signal $\x$ is assumed to be a Gaussian
$k$-sparse vector with support set $K$, where
$k=|K|=(1+\epsilon_0)\mu_{W}(\delta)n$, for some $\epsilon_0>0$. By a
Gaussian $k$-sparse vector, we mean one where the nonzero entries are
i.i.d. Gaussian (zero mean and unit variance, say). It should be noted that the Gaussian distribution is only considered as a standard choice. We later extend our analysis to other signal distributions. The
solution $\hat{\x}$ to the
$\ell_1$ minimization obtained in step 1 of Algorithm
\ref{alg:modmain} is in all likelihood a dense vector.  The set $L$, as
defined in the algorithm, is the
$k$-support set of $\hat{\x}$ (\emph{i.e.} $L=supp_k(\hat{\x})$). We show that for small enough
$\epsilon_0$, the intersection of  $L$ and $K$ is with high
probability very large, so that $L$ can be counted as a good
approximation to $K$ (Figure \ref{fig:sigmal}).

In order to find a decent lower bound on $|L\cap K|$, we point out three
separate facts and establish a connection between them. First, we
prove a general lemma that provides a lower bound on the quantity $|L\cap K|$ as a function of
$\|\x-\hat{\x}\|_1$. Then, we discuss a critical property of
$\ell_1$ minimization known as \emph{weak robustness} which helps provide an
upper bound on the quantity $\|\x-\hat{\x}\|_1$. The robustness result is due to Xu \emph{et al.} and was first proved in \cite{isitrobust}. However, we provide explicit scaling laws for the robustness of $\ell_1$ minimization beyond the implicit results of \cite{isitrobust}. Finally, we leverage some concentration results for order statistics to derive explicit formulae for the obtained bounds. These steps will be elaborated in the remainder of this section.

\begin{definition}
For a $k$-sparse signal $\x$, we define $W(\x,\lambda)$ to be the size of the largest subset of nonzero entries of $\x$ that has a $\ell_1$ norm less than or equal to $\lambda$, i.e.,
\beq
\nonumber W(\x,\lambda) \triangleq \max\{|S|~|~S\subseteq supp(\x),~\|\x_S\|_1\leq \lambda\}.
\eeq
\end{definition}
\noindent Note that $W(\x,\lambda)$ is increasing in $\lambda$.
\begin{lem}
Let $\x$ be a $k$-sparse vector and $\hat{\x}$ be another vector. Also, let $K$ be the support set of $\x$ and $L$ be the $k$-support set of $\hat{\x}$. Then
\beq
|K\cap L|\geq k-W(\x,\|\x-\hat{\x}\|_1).
\eeq
\label{lem:deviation thm}
\end{lem}
\begin{proof}
Let $x_i$ be the $i$th entry of $\x$ and  $\e^* = (e_1,e_2,\dots,e_n)^T$ be the solution to the following minimization problem:
\begin{align}
&\text{minimize }\|\e\|_1\nonumber \\
\text{s.t}. &\max|(\x+\e)_{K\setminus L}| \leq \min|(\x+\e)_{L}|,
\label{eq:min aux}
\end{align}
\noindent where $K\setminus L$ denotes the subset of the entries of $K$ that are not in $L$.  Note that the vector $\hat{\x}-\x$ satisfies the constraint of the minimization problem (\ref{eq:min aux}). This is because $\x+(\hat{\x}-\x) = \hat{\x}$ and $L$ is the $k$-support of $\hat{\x}$. Therefore every entry of $\hat{\x}$ outside the set $L$ is smaller in amplitude than every entry inside $L$. Therefore since $\e^*$ is the optimal solution of (\ref{eq:min aux}) we must have:
\beq
\|\e^*\|_1 \leq \|\hat{\x}-\x\|_1.
\label{eq:aux1}
\eeq
\noindent Let $a = \max|(\x+\e^*)_{K\setminus L}|$.  Then for each
$i\in K\setminus L$, using the triangular inequality we have
\beq
|x_i|-|e_i| \leq |x_i + e_i| \leq a,~\forall i\in K\setminus L,
\label{eq:aux-+}
\eeq
\noindent and so:
\beq
|e_i| \geq \max(|x_i|-a,0),~\forall i\in K\setminus L.
\label{eq:aux-+1}
\eeq
\noindent Therefore, by summing up the inequalities in (\ref{eq:aux-+1}) for $i\in K\setminus L$ we have
\beq
\|\e^*_{K\setminus L}\|_1 \geq \sum_{i\in K\setminus L, |x_i| >a}(|x_i|-a).
\eeq
On the other hand, for all $i\in L\setminus K$, we have $|e_i| > a$, and therefore:
\beq
\|\e^*_{L\setminus K}\|_1 \geq a|L\setminus K|.
\eeq
But $|L\setminus K|=|K\setminus L|$ and hence it follows that
\bea
\|\e^*\|_1 &\geq& \|\e^*_{L\setminus K}\|_1 + \|\e^*_{K\setminus L}\|_1 \nonumber \\
&\geq& a|K\setminus L| + \sum_{i\in K\setminus L, |x_i| >a}(|x_i|-a)  \nonumber \\
&\geq& \sum_{i\in K\setminus L}|x_i| = \|\x_{K\setminus L}\|_1.
\label{eq:aux2}
\eea
\noindent (\ref{eq:aux1}) and (\ref{eq:aux2}) together imply that $\|\x-\hat{\x}\|_1\geq \|\x_{K\setminus L}\|_1$, which by definition means that $W(\x,\|\x-\hat{\x}\|_1)\geq |K\setminus L|$.
\end{proof}
We now introduce the notion of weak robustness, which
allows us to bound $\|\x-\hat{\x}\|_1$, and has the following formal
definition \cite{isitrobust}.
\begin{definition}
Let the set $S\subset\{1,2,\cdots,n\}$ and the subvector $\x_S$  be
fixed. An approximation $\hat{\x}$ to $\x$ is called weakly robust with respect to the set $S$ if, for some $C_S >
1$, it holds that
\beq
\|(\x-\hat{\x})_S\|_1 \leq \frac{2C_S}{C_S-1}\|\x_{\overline{S}}\|_1,
\label{def1}
\eeq
\noindent and
\beq
\|\x_S\|-\|\hat{\x}_S\|\leq \frac{2}{C_S-1}\|\x_{\overline{S}}\|_1.
\label{def2}
\eeq
\noindent $C_S$ is called the robustness parameter of the considered approximation for the set $S$.
\end{definition}
\noindent The weak robustness notion allows us to bound the error
in $\|\x -\hat{\x}\|_1$ in the following way. If $\hat{\x}$ is a weakly robust approximation to $\x$ with respect to the set $S$ and parameter $C_S>1$, such that $\A\x = \A\hat{\x}$, and if the matrix $\A_S$
obtained by retaining only those columns of $\A$ that are
indexed by $S$ has full column rank, then the quantity
\beq
\nonumber
~\kappa = \max_{\A\w =0, \w\neq 0} \frac{\|\w_S\|_1}{\|\w_{\overline{S}}\|_1},
\eeq
must be finite, and one can conclude that
\beq
\|\x-\hat{\x}\|_1 \leq \frac{2C_S(1+\kappa)}{C_S-1}\|\x_{\overline{S}}\|_1.
\label{eq:robustness}
\eeq
\noindent This result is due to \cite{isitrobust}, where in addition it has been shown that for Gaussian
i.i.d. measurement matrices $\A$, the solution of $\ell_1$ minimization provides a weakly robust approximation with high probability. In other words, for a randomly chosen subset $S$ with $\frac{|S|}{n} < \mu_{W}(\delta)$, there exists a robustness factor $C>1$ as a function of $\frac{|S|}{n}$ for which
(\ref{def1}) and (\ref{def2}) hold with high probability for an arbitrary vector $\x$, where $\hat{\x}$ is the solution obtained by $\ell_1$ minimization. Now let $k_1 = (1-\epsilon_1)\mu_{W}(\delta)n$ for some small
$\epsilon_1>0$, and $K_1$ be the $k_1$-support set of $\x$, namely,
the set of the largest $k_1$ entries of $\x$ in magnitude. Based on
equation (\ref{eq:robustness}) we may write
\beq
\|\x-\hat{\x}\|_1 \leq \frac{2C(\epsilon_1)(1+\kappa)}{C(\epsilon_1)-1}\|\x_{\overline{K_1}}\|_1,
\label{eq:robustness1}
\eeq
\noindent where for a fixed value of $\delta$, we have emphasized that the constant $C$ for the set $K_1$ is a function of $\epsilon_1$. Furthermore, $C(\epsilon_1)$ becomes arbitrarily close to $1$ as
$\epsilon_1\rightarrow 0$. $\kappa$ is also a bounded function of
$\epsilon_1$ and therefore we may replace it with an upper bound
$\kappa^*$. This provides a bound on $\|\x-\hat{\x}\|_1$. To explore
this inequality and understand its asymptotic behavior, we
apply a third result, which is a certain concentration bound on
the order statistics of Gaussian random variables.
\begin{lem}
Suppose $X_1,X_2,\cdots,X_N$ are $N$ i.i.d. $\mathcal{N}(0,1)$ random
variables. Let $S_N = \sum_{i=1}^{N}|X_i|$ and let $S_M$ be the sum of
the largest $M$ numbers among the $|X_i|$'s, for each
$1\leq M < N$. Then for every $\epsilon>0$ sufficiently small, as $N\rightarrow\infty$, if the ratio $M/N$ is kept constant, we have
\bea
\Prob\left(\left|\frac{S_N}{N} - \sqrt{\frac{2}{\pi}}\right| > \epsilon\right)\rightarrow 0, \label{eq:Order_stat1}\\
\Prob\left(\left|\frac{S_M}{S_N} -
\exp(-\frac{\Psi^2(\frac{M}{2N})}{2})\right|>\epsilon \right)\rightarrow 0, \label{eq:Order_stat2}
\eea
\noindent where $\Psi(x) = Q^{-1}(x)$ with $Q(x) = \frac{1}{\sqrt{2\pi}}\int_{x}^{\infty}e^{-\frac{y^2}{2}}dy$.
\label{lemma:Gaussian_Base}
\end{lem}
To make the proof more understandable and the paper more readable, we mention the general idea of the proof of the above lemma very coarsely here. The detailed proof is outlined in Appendix \ref{app:proof_of_order_stat}. For a particular instance of $X_1,\dots,X_N$, if $0<a<1$ is such that exactly a fraction $M/N$ of $|X_i|$'s are larger than $a$, then every $|X_i|$ which is larger than $a$ contributes to the sum $S_M$. Therefore $S_M$ can be thought of as those $|X_i|$'s that are larger than $a$. This can be expressed in another way. Let $\hat{X}_i$ be a random variable which is equal to $|X_i|$ if $|X_i|>a$ and is $0$ otherwise. We therefore conclude that $S_M$ is equal to the sum of $\sum_{i=1}^n{\hat{X_i}}$. Furthermore, when $N$ is large, it can be shown using concentration lemmas that $a$ will be arbitrarily close to the fixed number $\Psi(\frac{M}{2N})$, and thus the distributions of $\hat{X_i}$'s converge to the same distribution, namely the truncated absolute value of a normal distribution. Besides, when $a$ is constant $\hat{X_i}$'s are independent and therefore one can apply the law of large numbers to conclude that $S_M/S_N\approx \stexp \hat{X}_1/\stexp |X_1|$, which is the desired conclusion. These arguments are rigorously outlined in  Appendix \ref{app:proof_of_order_stat}.

Recall that we assumed that $\x$ is a $k$-sparse random Gaussian signal with $k=(1+\epsilon_0)\mu_W(\delta)n$, and we defined $K_1$ to be the $k_1$-support of $\x$, where $k_1 = (1-\epsilon_1)\mu_W(\delta)n$. We denoted by $K$ the support set of $\x$. Also, if $\hat{\x}$ is the approximation to $\x$ obtained by $\ell_1$ minimization, we denoted by $L$ the $k$-support set of $\hat{\x}$. As a direct consequence of Lemma \ref{lemma:Gaussian_Base} we can write:
\beq
\Prob\left(\left|\frac{\|\x_{\overline{K_1}}\|_1}{\|\x\|_1} - (1-e^{-0.5\Psi^2(0.5\frac{1-\epsilon_1}{1+\epsilon_0})})\right|>\epsilon\right)\rightarrow 0,
\label{eq:kbar bound}
\eeq
\noindent for $\epsilon>0$ sufficiently small as $n\rightarrow\infty$. Define
\beq
\zeta(\epsilon_0) \triangleq \inf_{\epsilon_1>0}\frac{2C(\epsilon_1)(1+\kappa^*)}{C(\epsilon_1)-1}(1-e^{-0.5\Psi^2(0.5\frac{1-\epsilon_1}{1+\epsilon_0})}).
\label{eq:zeta}
\eeq

\noindent Incorporating (\ref{eq:robustness1}) into (\ref{eq:kbar bound}) we may write
\beq
\Prob\left(\frac{\|\x-\hat{\x}\|_1}{\|\x\|_1} - \zeta(\epsilon_0) < \epsilon\right)\rightarrow 1,
\label{eq:robustness2}
\eeq
\noindent for $\epsilon>0$ sufficiently small as $n\rightarrow\infty$.

Let us summarize our conclusions so far. First, we were able to show that
$|K\cap L|\geq k-W(\x,\|\x-\hat{\x}\|_1)$. The weak
robustness of $\ell_1$ minimization and the Gaussianity of the signal then
led us to the fact that for large $n$ with high probability $\|\x-\hat{\x}\|_1
\leq \zeta(\epsilon_0)\|\x\|_1$. These results build up the next key
theorem, which is the conclusion of this section.
\begin{thm}[Approximate Support Recovery]
Let $\A$ be an i.i.d. Gaussian $m\times n$ measurement matrix with
$\frac{m}{n}=\delta$. Let $k=(1+\epsilon_0)\mu_{W}(\delta)$ and $\x$
be an $n\times1$ random Gaussian $k$-sparse  signal. Suppose that
$\hat{\x}$ is the approximation to $\x$ given by $\ell_1$ minimization, i.e. $\hat{\x}=argmin_{\A\z=\A\x}\|\z\|_1$. Then, as
$n\rightarrow\infty$, for all $\epsilon>0$,
\beq
\small
\Prob\left(\frac{|supp(\x) \cap supp_k(\hat{\x})|}{k} -
2Q(\sqrt{-2\log(1-\zeta(\epsilon_0))})>-\epsilon\right)\rightarrow 1,
\label{eq:support recovery}
\eeq
\noindent where $\zeta(\cdot)$ is defined in (\ref{eq:zeta}).
\label{thm:l_1 support recovery}
\end{thm}

Before proving the above theorem, we mention the following useful lemma, the proof of which will be given in Appendix \ref{app:proof_of_W(x,ax)}.
\begin{lem}
Let $\x$ be a random $k$-sparse Gaussian vector of size $n$, and $0<\alpha<1$. For any positive $\epsilon$, the following happens with high probability as $n,k\rightarrow \infty$:
\beq
\frac{W(\x,\alpha\|\x\|_1)}{k}<(1-2Q(\sqrt{-2\log(1-\alpha)}))+\epsilon.
\eeq
\label{lem:W(x,a)}
\end{lem}
\begin{proof}[Proof of Theorem \ref{thm:l_1 support recovery}]
From equation (\ref{eq:robustness2}), for every $\epsilon'>0$ and large enough $n$, with high probability we have
$\|\x-\hat{\x}\|_1<(\zeta(\epsilon_0)+\epsilon')\|\x\|_1$. Therefore,
from Lemma \ref{lem:deviation thm} and the fact that $W(\x,\lambda)$
is increasing in $\lambda$, $|K\cap L| \geq k -
W(\x,(\zeta(\epsilon_0)+\epsilon')\|\x\|_1)$ with high
probability. Replacing for $W(\x,(\zeta(\epsilon_0)+\epsilon'))$ with the upper bound given by  Lemma \ref{lem:W(x,a)}, it follows that with very high
probability $\frac{|K\cap L|}{k}\geq
2Q(\sqrt{-2\log(1-\zeta(\epsilon_0)-\epsilon')})-\epsilon''$. We can now let $\epsilon'$ go to zero and the proof is completed.
\end{proof}
\noindent Note that if $\lim_{\epsilon_0\rightarrow0}\zeta(\epsilon_0)=0$, then  Theorem \ref{thm:l_1 support recovery} implies that $\frac{|K\cap L|}{k}$ becomes arbitrarily close to 1, which means that using $\ell_1$ minimization it is possible to closely estimate the support set of $\x$. We show in the sequel that this is in fact the case.

\subsection{Scaling Law of $\ell_1$ Minimization}
In order to show that the robust approximation of the sparse signal at step 1 of Algorithm \ref{alg:modmain} leads to perfect recovery at step 3, we need to obtain an explicit bound for the term $\zeta(\epsilon_0)$. This in turn requires calculating a solid relationship between the robustness parameter $C(\epsilon_1)$, and the back-off fraction $\epsilon_1$. For i.i.d. Gaussian matrices, we derive an explicit lower bound on $C(\epsilon_1)$ as a function of $\epsilon_1$ through the following theorem, the proof of which appears in Appendix \ref{app:proof_of_scaling}.

\begin{thm}[Scaling law of $\ell_1$ minimization for Gaussians.]
Let $\A$ be an $m\times n$ i.i.d. Gaussian matrix with $m=\delta n$, and $\mu_{W}(\delta)$ be the weak recovery threshold of $\ell_1$ minimization for $A$. For sufficiently large $n$, the (weak) robustness parameter $C(\epsilon_1)$ for a randomly chosen $k_1$-support $K_1$ of size $k_1=(1-\epsilon_1)\mu_{W}(\delta)n$  (see equation \ref{eq:robustness1}) satisfies:
\beq
C(\epsilon_1) \geq \frac{1}{\sqrt{1-\epsilon_1}}.
\label{eq:C_vs_eps}
\eeq
\label{thm:scale}
\end{thm}

We now derive an asymptotic upper bound on the term $\zeta(\epsilon_0)$ using the above relationship. Replacing the bound of (\ref{eq:C_vs_eps}) in the definition of $\zeta(\epsilon_0)$, we obtain:
\bea
\nonumber \zeta(\epsilon_0) &=& \inf_{\epsilon_1>0}\frac{2C(\epsilon_1)(1+\kappa^*)}{C(\epsilon_1)-1}\left(1-e^{-0.5\Psi^2(0.5\frac{1-\epsilon_1}{1+\epsilon_0})}\right) \\
&\leq& \inf_{\epsilon_1>0}\frac{2(1+\kappa^*)}{1-\sqrt{1-\epsilon_1}}\left(1-e^{-0.5\Psi^2(0.5\frac{1-\epsilon_1}{1+\epsilon_0})}\right)  \\
&\leq& \frac{4(1+\kappa^*)}{\epsilon_0}\left(1-e^{-0.5\Psi^2(0.5\frac{1-\epsilon_0}{1+\epsilon_0})}\right),
\label{eq:zbound2}
\eea
\noindent where (\ref{eq:zbound2}) is obtained by simply taking $\epsilon_1=\epsilon_0$, and using the fact that $\frac{1}{1-\sqrt{1-\epsilon_0}}\leq 2/\epsilon_0$. We use the Taylor approximation of the inverse error function to bound the right hand side of (\ref{eq:zbound2}). Note that:
\bea
\Psi(0.5\frac{1-\epsilon_0}{1+\epsilon_0}) &=& \sqrt{2}\cdot\text{erf}^{-1}\left(\frac{2\epsilon_0}{1+\epsilon_0}\right)\\
&=&  \sqrt{2\pi}\cdot\epsilon_0 + \text{\textit{o}}(\epsilon_0^2).
\eea
\noindent It follows that:
\beq
\zeta(\epsilon_0) \leq 4\pi(1+\kappa^*)\epsilon_0 + \mathcal{O}(\epsilon_0^2),
\eeq
\noindent As $\epsilon_0\rightarrow 0$. Therefore, we can immediately see that $\lim_{\epsilon_0\rightarrow0}\zeta(\epsilon_0)=0$.

 \section{Perfect Recovery, Step 3 of the Algorithm}
\label{sec:perfect recovery}
In Section \ref{sec:robustness} we showed that if
$\epsilon_0$ is small, the $k$-support of $\hat{\x}$, namely
$L=supp_k(\hat{\x})$, has a significant overlap with the true support of
$\x$. We even found a quantitative lower bound on the size of this overlap
in Theorem \ref{thm:l_1 support recovery}. In step 3 of
Algorithm \ref{alg:modmain}, weighted $\ell_1$ minimization is used,
where the entries in  $\overline{L}$ are assigned a higher weight
than those in $L$. In \cite{Khajehnejad_weighted}, we have been able to
analyze the performance of such weighted $\ell_1$ minimization
algorithms. The idea is that if a sparse vector $\x$ can
be partitioned into two sets $L$ and $\overline{L}$, where in one set
the fraction of non-zeros is much larger than in the other set, then
(\ref{eq:weighted l_1}) can potentially recover $\x$ with an
appropriate choice of the weight $\omega > 1$, even though
$\ell_1$ minimization cannot.  The following theorem can be deduced
from the computations of \cite{journalweighted}.
\begin{thm}
Let $L\subset \{1,2,\cdots,n\}$ , $\omega>1$ and the fractions
$f_1,f_2\in[0,1]$ be given. Let $\gamma_1 = \frac{|L|}{n}$ and
$\gamma_2=1-\gamma_1$. There exists a threshold
$\lambda_c(\gamma_1,\gamma_2,f_1,f_2,\omega)$ such that with high
probability, almost all random sparse vectors $\x$ with \emph{at least}
$f_1\gamma_1n$ nonzero entries over the set $L$, and \emph{at most}
$f_2\gamma_2n$ nonzero entries over the set $\overline{L}$ can be
perfectly recovered using
$\min_{\A\z=\A\x}\|\z_L\|_1+\omega\|\z_{\overline{L}}\|_1$, where $\A$
is a $\lambda_c n\times n$ matrix with i.i.d. Gaussian entries.
\label{thm:delta}
\end{thm}
\noindent For completeness, in Appendix \ref{App:delta_c}, we provide
the calculation of $\lambda_c(\gamma_1,\gamma_2,f_1,f_2,\omega)$, based on the calculations of \cite{journalweighted}. A software package for computing such thresholds can also be found in
\cite{sotware link}.

\begin{proof}[Proof of Theorem \ref{thm: final thm}]
Recall that the solution of $\ell_1$ minimization in the first stage of Algorithm (\ref{alg:modmain}) is the vector $\hat{\x}$. We denoted by $L$ the $k$-support set of $\hat{\x}$, and by $L^c$ its complement set. The last stage of the algorithm is a weighted $\ell_1$ minimization that puts more weight on the entries of $\x$ outside the set $L$. The justification for this is the fact that the fraction of the nonzero entries of the target signal $\x$ over the set $L$ is supposedly larger than the fraction of the nonzero entries over $L^c$. Let us denote these fractions by $f_1$ and $f_2$ respectively, namely  $f_1=\frac{|L\cap K|}{|L|}$ and $f_2=\frac{|\overline{L}\cap
K|}{|\overline{L}|}$, where $K$ is the support of the target signal, unknown to the algorithm before running the weighted $\ell_1$ minimization of the last stage. Since we are using a weighted $\ell_1$ minimization, $\x$ will be recovered perfectly with high probability if the number of measurements is large than the threshold of weighted $\ell_1$ minimization for the nonuniform sparsity model of the target signal, namely if:
\beq
\lambda_c(\frac{k}{n},1-\frac{k}{n},f_1,f_2,\omega) \leq \delta,
\eeq
\noindent where $\lambda_c$ was defined in Theorem \ref{thm:delta} and was characterized in \cite{journalweighted}. On the other hand, through Theorem \ref{thm:l_1 support recovery}, we provided a lower bound on $f_1$ (and consequently an upper bound on $f_2$) and we showed that as $\epsilon_0\rightarrow 0$, $f_1$ converges to 1 (and consequently $f_2$ approaches zero). The asymptotic value of  $\lambda_c(\frac{k}{n},1-\frac{k}{n},f_1,f_2,\omega)$ will therefore be
equal to $\lambda_c(\mu_{W}(\delta),1-\mu_{W}(\delta),1,0,\omega)$, as $\epsilon\rightarrow 0$ (Recall that $k = (1+\epsilon_0)\mu_W(\delta)n$). Furthermore, from the computations of \cite{journalweighted}, it can be shown that $\lambda_c(\mu_{W}(\delta),1-\mu_{W}(\delta),1,0,\omega) < \delta$ for an appropriate choice of $\omega >1$, and that for a fixed $\omega$, the function $\lambda_c(\gamma_1,\gamma_2,f_1,f_2)$ is a continuous function of $\gamma_1$, $f_1$ and $f_2$. Furthermore, $k$, the lower bound on $f_1$ and the upper bound on $f_2$ obtained from Theorem \ref{thm:l_1 support recovery} are all continuous functions of $\epsilon_0$ in this case. Therefore, we can conclude that for a strictly positive $\epsilon_0$ and corresponding overlap fractions $f_1$ and
$f_2$, $\lambda_c((1+\epsilon_0)\mu_{W}(\delta),1-(1+\epsilon_0)\mu_{W}
(\delta),f_1,f_2,\omega)< \delta$. This means that for some strictly positive $\epsilon_0$ the number of measurements that is required to reconstruct the signal precisely in the last stage of the algorithm is less than the number of measurements in $\A$, i.e. $\x$ will be recovered with high probability, despite the fact that it has more nonzero entries that the weak threshold of $\ell_1$ minimization. This completes the proof.
\end{proof}
\section{Generalization to Beyond Gaussians}
\label{sec:generalization}
The theoretical threshold improvement of the proposed iterative $\ell_1$ minimization algorithm was demonstrated for the case of i.i.d. Gaussian matrices, and sparse vectors with independent Gaussian nonzero entries.  It is reasonable to ask if we can extend these results to sparse signals with other distributions. We address this problem in this section. In summary, we prove that the theoretical threshold improvement can be generalized to sparse signals whose nonzero entries obey a more general class of distributions, namely continuous symmetric distributions with a non-vanishing finite order derivative at the origin.  This is outlined in the following section.

\subsection{Arbitrary Distributions}
The attentive reader will note that the only step where we used the Gaussianity of
the signal in the proof of threshold improvement was in the the order statistics results of Lemma
\ref{lemma:Gaussian_Base}. This result has the following
interpretation. For $N$ i.i.d. random variables, the
ratio $\frac{S_M}{S_N}$ can be approximated by a known function of $\frac{M}{N}$. In the Gaussian
case, this function behaves as $1-(1-\frac{M}{N})^2$, as $M\rightarrow
N$. For constant magnitude signals (say BPSK), the function behaves as
$\frac{M}{N}$, for $M\rightarrow N$, which predicts that the
reweighted method yields no improvement. A more careful analysis reveals that the improvement
over $\ell_1$ minimization depends on the behavior of
$\frac{S_M}{S_N}$, as $M\rightarrow N$, which in term depends on the
smallest order $n$ for which $f^{(n)}(0)\neq 0$, i.e., the smallest
$n$ such that the $n$-th derivative of the distribution at the origin is nonzero. We formalize these results by generalizing the arguments of the previous section. First, we present a generalization of Lemma \ref{lemma:Gaussian_Base}  for arbitrary symmetric distributions.

\begin{lem}
Suppose $X,X_1,X_2,\cdots,X_n$ are $N$ i.i.d. random
variables, drawn from a symmetric distribution  $f(\cdot)$. Let $S_N = \sum_{i=1}^{N}|X_i|$ and let $S_M$ be the sum of the largest $M$ numbers among $|X_i|$'s, for each
$1\leq M < N$. If $f(\cdot)$ is integrable, and if for every finite $a>0$, the integral $\int_{a}^{\infty}x^2f(x)dx$ is finite, then for every $\epsilon>0$ sufficiently small, as $N\rightarrow\infty$ and the ratio $M/N$ is kept constant, the following holds
\bea
\Prob\left(\left|\frac{S_M}{S_N} - (1-2\frac{\int_{0}^{\Psi_f(\frac{M}{2N})}x\cdot f(x) dx}{\stexp_{f(\cdot)}|X|})\right|>\epsilon \right)\rightarrow 0,
\eea
\noindent where $\Psi_f(x) = Q_f^{-1}(x)$ with $Q_f(x) =\int_{x}^{\infty}f(y)dy$.
\label{lemma:arbitrary_concentration}
\end{lem}
Using the above lemma, we can modify the concentration term of equation (\ref{eq:kbar bound}) for the term $\frac{\|\x_{\overline{K_1}}\|_1}{\|\x\|_1}$, where the distribution of the nonzero entries of $\x$ is $f(\cdot)$. The resulting concentration thus becomes:
\beq
\Prob\left(\left|\frac{\|\x_{\overline{K_1}}\|_1}{\|\x\|_1} - 2\frac{\int_{0}^{\Psi_f(\frac{(1-\epsilon_1)}{2(1+\epsilon_0)})}x\cdot f(x) dx}{\stexp_{f(\cdot)}|X|}\right|>\epsilon\right)\rightarrow0,
\label{eq:kbar bound_gen}
\eeq
\noindent which, when put together with the bound of (\ref{eq:robustness1}) results in (Note that the bound in (\ref{eq:robustness1}) is independent from the distribution of $\x$):
\beq
\Prob\left(\frac{\|\x-\hat{\x}\|_1}{\|\x\|_1} - \zeta_f(\epsilon_0) < \epsilon\right)\rightarrow 1,
\label{eq:robustness2}
\eeq
\noindent for every $\epsilon>0$. Here $\zeta_f(\epsilon_0)$ is defined by:

\beq
\zeta_f(\epsilon_0) \triangleq \inf_{\epsilon_1>0}\frac{2C(\epsilon_1)(1+\kappa^*)}{C(\epsilon_1)-1}\times  2\frac{\int_{0}^{\Psi_f(\frac{(1-\epsilon_1)}{2(1+\epsilon_0)})}x\cdot f(x) dx}{\stexp_{f(\cdot)}|X|}.
\label{eq:zeta_f}
\eeq

\noindent Consequently, following similar arguments as in the proofs of Theorem \ref{thm:l_1 support recovery}, we can state the following theorem as a generalization of the approximate support recovery of $\ell_1$ minimization for arbitrary distributions, the proof of which is immediate.

\begin{thm}[Approximate Support Recovery/Generalization]
Let $\A$ be an i.i.d. Gaussian $m\times n$ measurement matrix with
$\frac{m}{n}=\delta$. Let $k=(1+\epsilon_0)\mu_{W}(\delta)n$ and $\x$
be an $n\times1$ $k$-sparse  signal whose nonzero entries are independently drawn from a distribution $f(\cdot)$ which satisfies the conditions of Lemma \ref{lemma:arbitrary_concentration}. Suppose that
$\hat{\x}$ is the approximation to $\x$ given by the $\ell_1$ minimization, i.e. $\hat{\x}=argmin_{\A\z=\A\x}\|\z\|_1$. Then, as
$n\rightarrow\infty$, for $\epsilon>0$ sufficiently small, we have
\beq
\small
\Prob\left(\frac{|supp(\x) \cap supp_k(\hat{\x})|}{k} -
2Q_f(\sqrt{-2\log(1-\zeta_f(\epsilon_0))})>-\epsilon\right)\rightarrow 1,
\label{eq:support recovery}
\eeq
\noindent where $\zeta_f(\cdot)$ is defined in (\ref{eq:zeta_f}).
\label{thm:l_1 support recovery_generalize}
\end{thm}

Note that $Q_f(\cdot)$ is always a decreasing function which is equal to zero at the origin for symmetric distributions. Therefore, the overlap fraction given by Theorem \ref{thm:l_1 support recovery_generalize} can be arbitrarily close to 1, provided that $\zeta_f(\epsilon_0)$ is sufficiently small. Therefore, the key in further conclusions on the above bound is to derive a bound on the term $\zeta_f(\epsilon_0)$, and show that it becomes arbitrarily small. For BPSK signals for instance, the term $\frac{\|\x_{\overline{K_1}}\|_1}{\|\x\|_1}$ is always equal to $\epsilon_0$, and therefore we cannot guarantee that $\zeta(\epsilon_0)$ vanishes asymptotically as $\epsilon_0\rightarrow 0$ based on (\ref{eq:zeta_f}). In fact we prove that $\lim_{\epsilon_0\rightarrow 0}\zeta(\epsilon_0) = 0$,  for distributions $f(\cdot)$ for which one of the finite order derivatives at the origin is nonzero, stated formally in the following lemma:
\begin{lem}
Let $f(\cdot)$ be a symmetric distribution which satisfies the conditions of Lemma \ref{lemma:arbitrary_concentration}. If for some integer $r\geq 0$, the $r$'th order derivative of $f(\cdot)$ at origin exists and does not vanish, i.e., $f^{(r)}(0)\neq 0$, then $\zeta_f(\epsilon_0)=\mathcal{O}(\epsilon_0^{1/(r+1)})$, as $\epsilon_0 \rightarrow 0$. Consequently, the support set approximation of $\ell_1$ minimization is asymptotically perfect with high probability as $\epsilon_0\rightarrow 0$.
\end{lem}
\begin{proof}
For simplicity, we take $\epsilon_1$ in the definition of $\zeta_f(\epsilon_0)$ to be equal to $\epsilon_0$, which only provides an upper bound. Since $f^{(r)}(0) > 0$ and $f(\cdot)$ is continuous, we conclude that for some constant $c>0$, and sufficiently small $x$, $f(x)\geq c\times x^r$. Therefore,
\beq
1/2-Q_f(x) =\int_{0}^{x}f(t)dt \geq \frac{c}{r+1}x^{r+1},
\eeq
\noindent and thus,
\beq
x\geq \Psi_f(1/2-\frac{c}{r+1}x^{r+1}),
\label{eq:x<Psi()}
\eeq
\noindent for sufficiently small $x$. Note that we have used the fact that $\Psi_f(\cdot)$ is a decreasing function.    Equivalently, (\ref{eq:x<Psi()}) means that
\beq
\Psi_f(1/2-x) =  \mathcal{O}(x^{1/(r+1)})
\label{eq:Psi=O}
\eeq
\noindent as $x\rightarrow 0$. On the other hand, note that $\frac{1-\epsilon_0}{2(1+\epsilon_0)} \geq 1/2-\epsilon_0$, and thus:
\beq  \Psi_f(\frac{1-\epsilon_0}{2(1+\epsilon_0)}) \leq \Psi_f(1/2-\epsilon_0).
\label{Psi<Psi}
\eeq
\noindent It follows from the above, (\ref{eq:Psi=O}), and the fact that $f(x) = \mathcal{O}(x^r)$ as $x\rightarrow 0$ that

\beq \int_{0}^{\Psi_f(\frac{(1-\epsilon_0)}{2(1+\epsilon_0)})}x\cdot f(x) dx = \mathcal{O}(\epsilon_0^{1+1/(r+1)}),\eeq

\noindent as $\epsilon_0\rightarrow 0$. Furthermore, from Theorem \ref{thm:scale}, we know that $C(\epsilon_1)\geq 1/\sqrt{1-\epsilon_0}$ (note that $\epsilon_1 = \epsilon_0$), and therefore $\frac{2C(\epsilon_1)(1+\kappa^*)}{C(\epsilon_1)-1} = \mathcal{O}(1/\epsilon_0)$ as $\epsilon_0\rightarrow 0$. Also, $\stexp_{f(\cdot)}|X|>0$ is constant. Therefore, from these conclusions and the definition of $\zeta_f(\cdot)$,  it follows that $\zeta_f(\epsilon_0) = \mathcal{O}(\epsilon_0^{1/(r+1)})$, as $\epsilon_0\rightarrow 0$.

\end{proof}

As a numerical example, we compute a theoretical bound for the approximate support recovery of $\ell_1$ minimization and threshold improvement in the case of $\delta= 0.5555$.  It is  easy to verify numerically that the conditions of Theorem \ref{thm: final thm} hold. The value of $\kappa^*$ is no more than $\sqrt{3}$ in this case. A theoretical bound on the overlap fraction between the $k$-support set of $\hat{\x}$ and the support set of the $k$-sparse $\x$ for an arbitrary distribution is provided by Theorem \ref{thm:l_1 support recovery_generalize}, where $k=(1+\epsilon_0)\mu_W(\delta)n$. We have computed this bound for three different distributions: Gaussian, uniform (-1,1) and a two sided Rayleigh distribution. The value of $r$, namely the smallest nonzero derivative order is $0$ for Gaussian and uniform  distributions, and is $1$ for the Rayleigh distribution. The computed bounds are plotted in Figure \ref{fig:overlaps_upperbounds}.  Furthermore, using  a value of $\omega=10$, and based on the premise of Theorem \ref{thm: final thm} and the computed bounds, we can certify an improvement of $\epsilon_0=5\times10^{-4}$ in the weak recovery threshold in the case of Gaussian distribution. For the uniform and Rayleigh distributions, the theoretical predictions in the improvement of recovery thresholds are smaller than the case of Gaussian, but are still strictly positive.  These improvement guarantees are of course much smaller than the practical values we would observe in practice, as will be illustrated in the following section.

\begin{figure}[t]
\centering
\includegraphics[width= 0.5\textwidth]{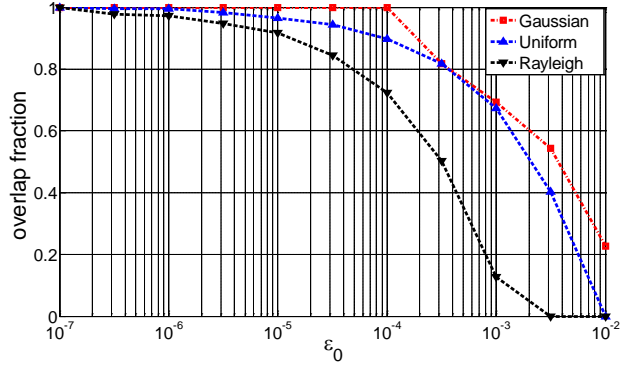}
  \caption{\scriptsize{Theoretical lower bound on the correct support estimation of $\ell_1$ minimization, as a function of the weak threshold exceeding fraction $\epsilon_0$. The plots are based on the theoretical results of Theorem \ref{thm:l_1 support recovery_generalize}, and are derived for Gaussian, uniform and two sided Rayleigh distributions.}}
  \label{fig:overlaps_upperbounds}
\end{figure}

\section{Simulations}
\label{sec:simulation}

We demonstrate the validity of the theoretical results of the previous sections, and the performance of Algorithm \ref{alg:modmain} by a few numerical simulations. The purpose of the simulations of this section is both to evaluate the performance of the proposed reweighted $\ell_1$ algorithm in practice, and to verify its distribution dependent behavior. Figure \ref{fig:simultions} shows the empirical performance of Algorithm \ref{alg:modmain} for sparse signals with various distributions. Here the signal
dimension is $n=200$, and the number of measurements is $m=112$, which
corresponds to a value of $\delta = 0.5555$. We generated random
sparse signals with i.i.d. entries coming from certain
distributions, namely Gaussian, uniform,  Rayleigh, square root of
$\chi$-square with 4 degrees of freedom and, square root of
$\chi$-square with 6 degrees of freedom. All of these distributions are continuous and have some finite-order non-vanishing derivative at the origin. In fact, in an increasing order of the mentioned distributions, the smallest order of nonzero derivative at the origin varies from 0 to 3. In other words, the pdf of a Gaussian and a uniform $(-1,1)$ distribution is nonzero at 0. The pdf of the Rayleigh distribution is zero at the origin, but has a nonzero derivative. Finally, the pdf's of square root of a $\chi$-square with 4 and 6 degrees of freedom have second and third nonzero derivatives at the origin, respectively. In Figure \ref{fig:simultions},  solid lines represent the simulation
results for ordinary $\ell_1$ minimization, and different colors
indicate different distributions. Dashed lines are used to show the
results for Algorithm \ref{alg:modmain}. Notice that the more derivatives that vanish at the
origin, the less significant improvement over $\ell_1$ minimization is observed, which is consistent with the analysis of Section \ref{sec:generalization}. The
Gaussian and uniform distributions are flat
and nonzero at the origin and show an impressive more than 20\%
improvement in the weak threshold (from 45 to 55 in this case).

\begin{figure}[t]
\centering
  \includegraphics[width= 0.5\textwidth]{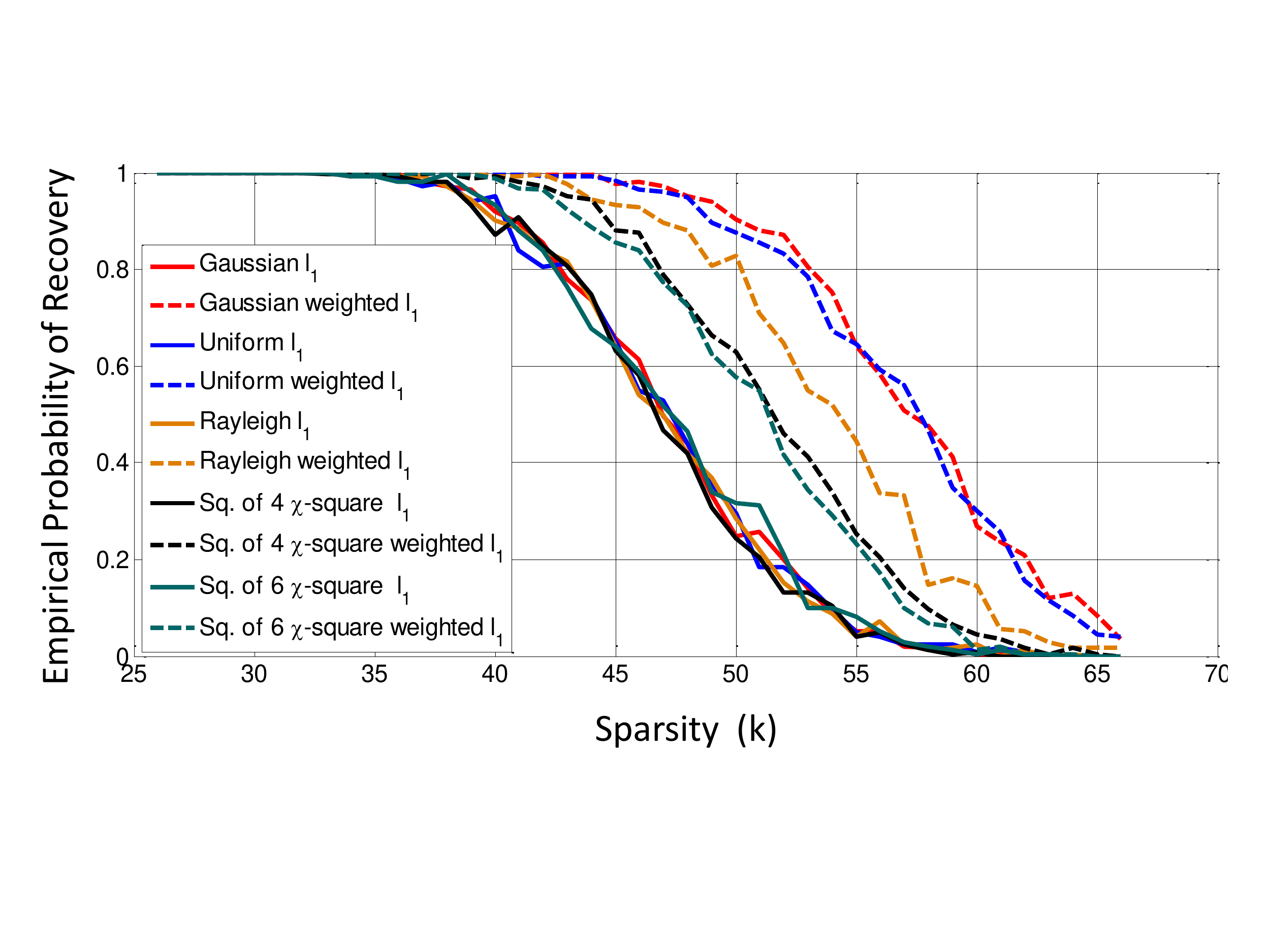}
  \caption{\scriptsize{Empirical Recovery Percentage for $n=200$ and $\delta = 0.5555$.}}
  \label{fig:simultions}
\end{figure}

In Figure \ref{fig:overlaps}, the overlap between the support set of a $k$-sparse signal $\x$ and the $k$-support set of the approximation $\hat{\x}$ given by $\ell_1$ minimization averaged over 400 random samples is plotted. Again, five different distributions were considered. It is apparent that overlap fraction is a decreasing function of $k$, and depends on the smoothness of the probability distribution at origin.
\begin{figure}[t]
\centering
  \includegraphics[width= 0.5\textwidth]{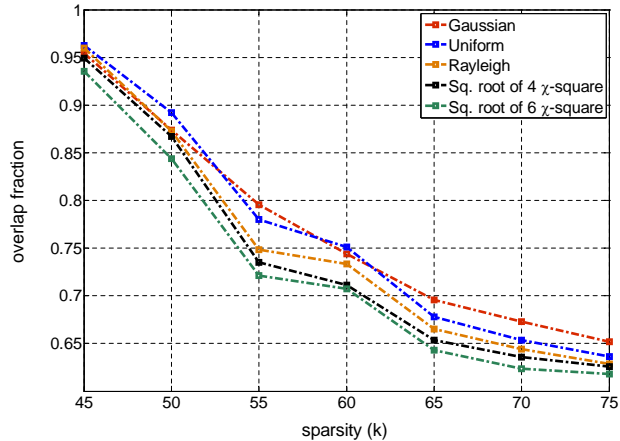}
  \caption{\scriptsize{Empirical overlap between the support set of a $k$-sparse vector and the $k$-support set of the $\ell_1$ optimum, for $n=200$ and $\delta = 0.5555$. Nonzero coefficients of signal are drawn from five different distributions (displayed). The average is over 400 samples.}}
  \label{fig:overlaps}
\end{figure}


%

We also report experimental results using regular $\ell_1$ and reweighted $\ell_1$ minimization  recovery algorithms over real world data. We have chosen a pair of satellite images (Figure \ref{fig:sat image})  taken in two different years, 1989 (left) and 2000 (right), from the New Britain rainforest in Papua New Guinea. Images originally belongs to Royal Society for the Protection of Birds and was taken from the Guardian archive, an article on deforestation. These images are generally recorded to evaluate environmental effects such as \emph{deforestation}.  The difference of images taken at different times is generally not very significant, and thus can be thought of as compressible. We have applied $\ell_1$ minimization  to recover the difference image over the subframe (subset of the original images) identified by the red rectangles in Figure \ref{fig:sat image}. In addition, we also implemented the reweighted $\ell_1$ minimization of Algorithm \ref{alg:modmain}, with $k=0.1n$ ($n$ being the total number of frame pixels), which assumed no prior knowledge about the structural sparsity of the signal or the nonzero coefficients. This value of $k$ was chosen heuristically, and is close to the actual support size of the signal. The original size of the image is $275\times 227$. We reduced the resolution by roughly a factor of $0.05$ for more tractability of $\ell_1$ solver in MATLAB. In addition, only the gray scale version of the difference image was taken into account, and was normalized so that the maximum intensity is 1. Furthermore, prior to compression, the difference image was further sparsified by rounding the intensities less than 0.1 to zero.  We pick the weight value $\omega = 2$ for the weighting stage of the reweighted $\ell_1$ algorithms. The normalized recovery error is defined to be the sum square of the intensity differences in the recovered and the original image, divided by the sum square of the original image intensity, i.e. $\sum_{i\in\text{frame}}(I_i-\hat{I}_i)^2 / \sum_{i\in\text{frame}}I_i^2$.  The average normalized error for $\ell_1$ minimization and reweighted $\ell_1 $ minimization is displayed in Figure \ref{fig:sat} as a function of $\delta$. The average is taken over $50$ realizations of i.i.d. Gaussian measurement matrices for each $\delta$. As can be seen, the recovery improvement is significant in the reweighted $\ell_1$ minimization.

\begin{figure}[t]
\centering
  \includegraphics[width= 0.33\textwidth]{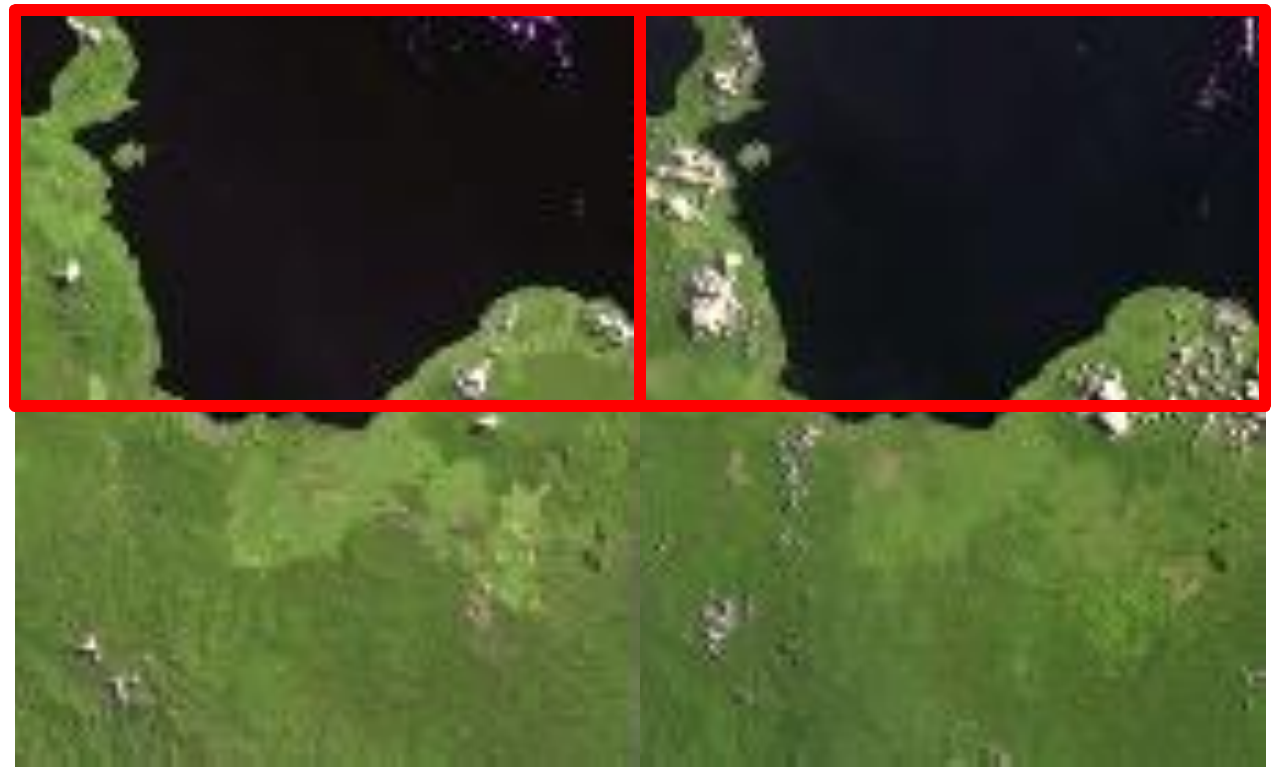}
  \caption{ \scriptsize{Satellite images taken from the New Britain rainforest in Papua Guina at 1989 (left) and 2000 (right). Red boxes identify the subframe used for the experiment, and green boxes identify the regions with higher associated weight in the weighted $\ell_1$ recovery. Image originally belongs to Royal Society for the Protection of Birds and was taken from the Guardian archive, an article on deforestation \url{http://www.guardian.co.uk/environment/2008/jan/09/endangeredspecies.endangeredhabitats}.}}
  \label{fig:sat image}
\end{figure}

\begin{figure}[t]
\centering
  \includegraphics[width= 0.33\textwidth]{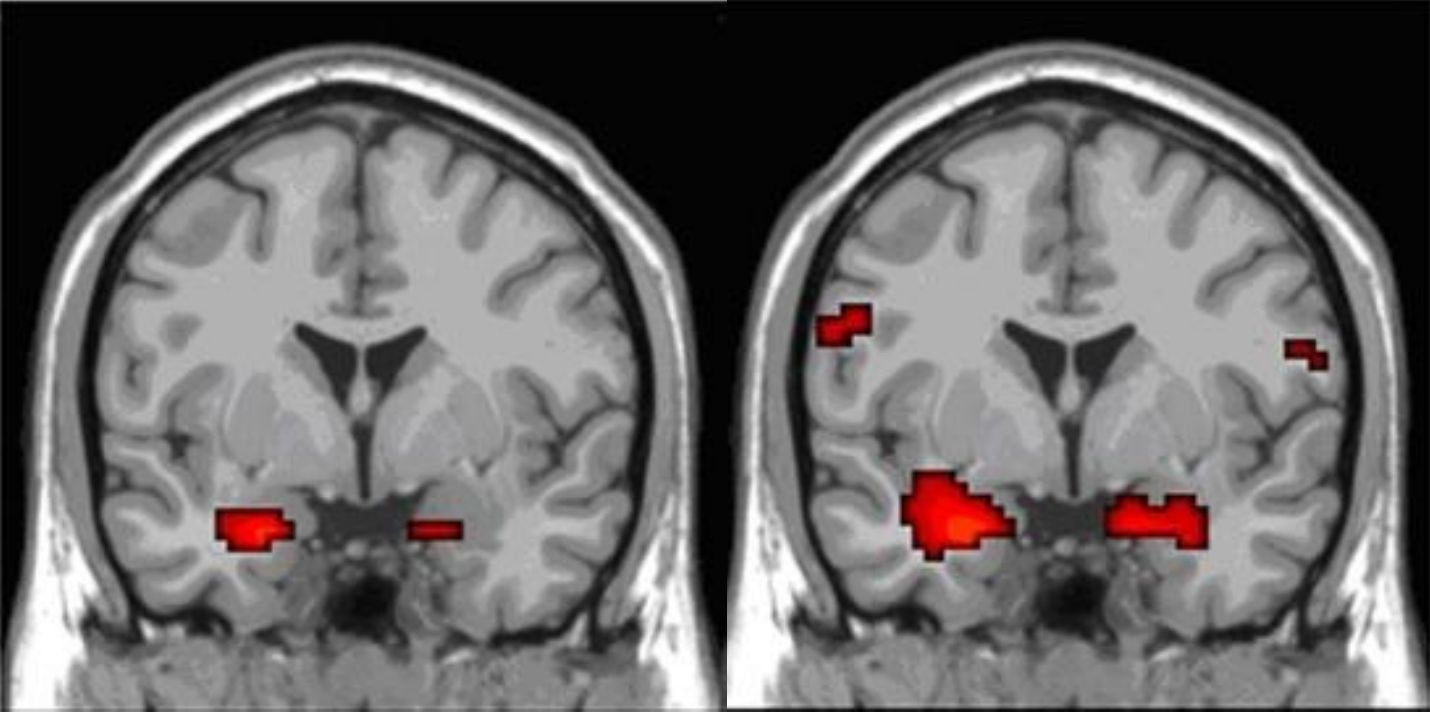}
  \caption{ \scriptsize{Functional MRI images of the brain at two different instances illustrating the brain activity. Green boxes identify the region with higher associated weight in the weighted $\ell_1$ recovery.
  Image is adopted from \url{https://sites.google.com/site/psychopharmacology2010/student-wiki-for-quiz-9}.}}
  \label{fig:brain image}
\end{figure}

\begin{figure}[t]
\centering
    \subfloat[]{\label{fig:sat}\includegraphics[width= 0.4\textwidth]{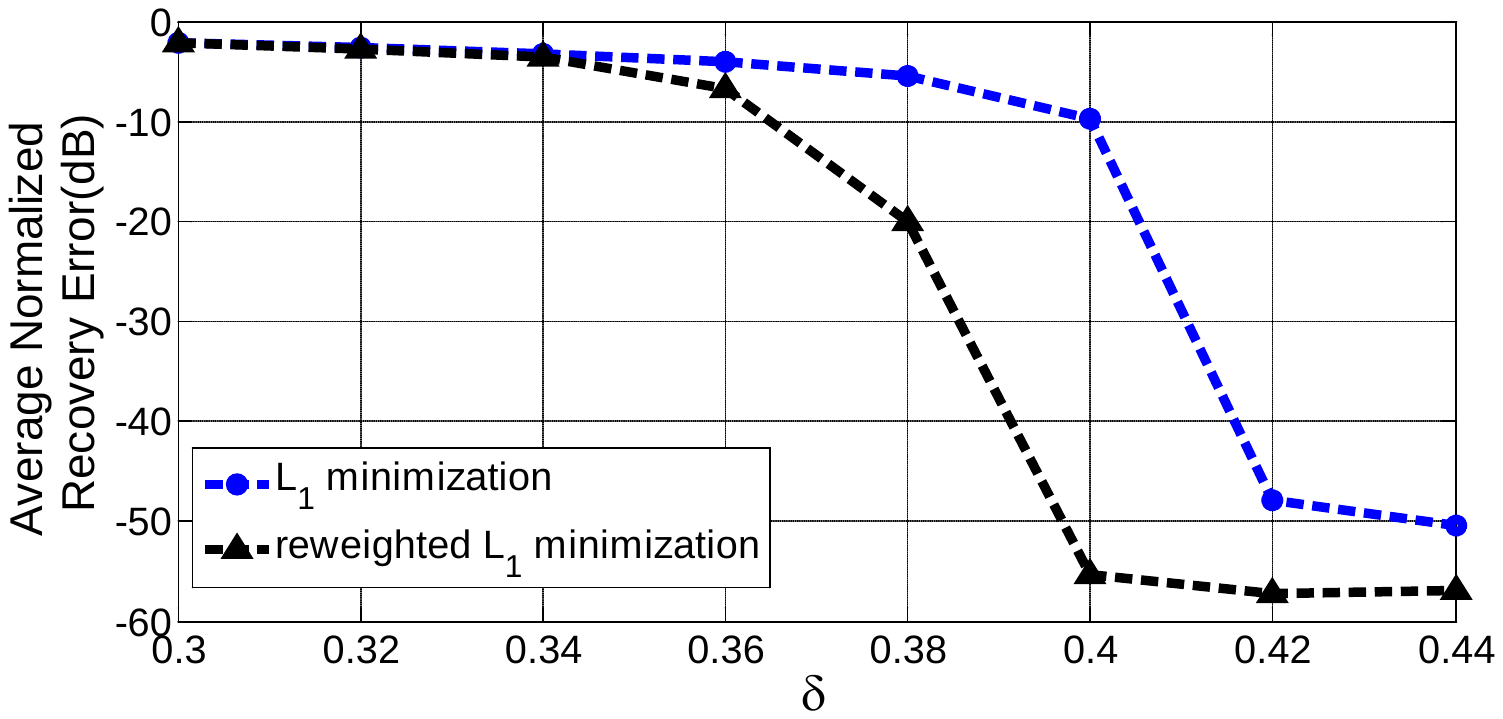}}
  \subfloat[]{\label{fig:brain}\includegraphics[width= 0.4\textwidth]{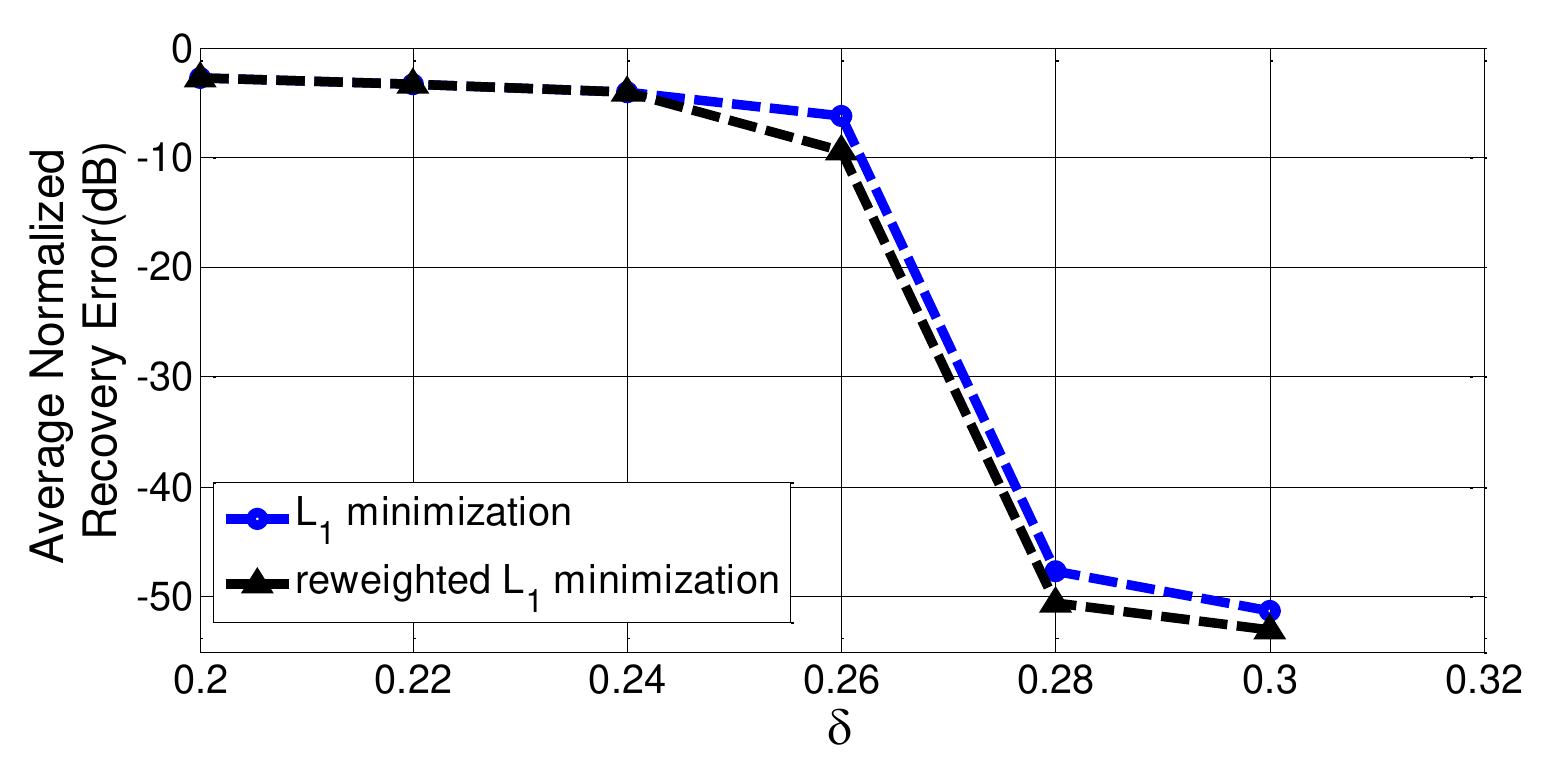}}
  \caption{ \scriptsize{Average normalized recovery error for $\ell_1$, and reweighted $\ell_1$ minimization recovery of the difference between the subframes of (a) a pair of satellite images shown in Figure \ref{fig:sat image}, and (b) the pair of brain fMRI images shown in Figure \ref{fig:brain image}.  Data is averaged over different realizations of measurement matrices for each $\delta$. }}
  \label{fig:simultions satellite_brain}
\end{figure}

Another experiment was done on a pair of brain fMRI images taken at two different instances of time, shown in Figure \ref{fig:brain image}. Similar to the satellite images, the objective is to recover the difference image from a set of compressed measurements The original image size is $271\times 271$, and similar preprocessing steps as for the satellite images were done before compression.  We used $\ell_1$ minimization and Algorithm \ref{alg:modmain} with no presumed prior information, with  $k=0.1n$ and $\omega = 1.3$. The average normalized recovery errors are displayed in Figure \ref{fig:brain}, from which we can infer similar conclusions as in the case of satellite images.

\section{Conclusion}
We introduced a new two-step reweighted $\ell_1$ minimization for the recovery of linearly compressed sparse signals. We proved that for sparse signals the nonzero entries of which are drawn from a broad class of continuous distributions, the proposed algorithm achieves a recovery threshold strictly better than that of $\ell_1$ minimization. Our theoretical analysis predicts that the performance improvement strongly depends on the distribution of the nonzero entries, and should be better for distributions with a smaller non-vanishing order of derivative at the origin. This was very closely verified by our numerical simulations. For distributions with no finite order non-vanishing derivative at origin, our analysis does not predict any improvement in the performance. This is also the case in practice: For ternary signals with nonzero values equal to $\pm1$ no improvement is observed in the empirical recovery threshold over the regular $\ell_1$ minimization. Our analysis was based on random Gaussian measurement matrices, and the robustness results of $\ell_1$ minimization. Possible related future research could address other measurement matrix ensembles, and the development of reweighted algorithms that can universally improve the recovery performance of linear programming. On the other hand, the improvement predictions using our theoretical tools are not tight, due to upper bounding techniques and worst case considerations in various parts of our proofs, specially in predicting the approximate support recovery potential of $\ell_1$ minimization. Future work can concentrate on tightening these bounds through more clever techniques, and consequently achieving more promising performance guarantees for reweighted linear programming.

\appendix{}

\section{Proof of Lemma \ref{lemma:Gaussian_Base}}
\label{app:proof_of_order_stat}
Let $a = \Psi(\frac{M}{2N})$.  We consider random variables $\hat{X}_i = |X_i|\cdot {\bf{1}}\left(|X_i|>a\right)$ for each $1\leq i\leq N$, where ${\bf{1}}\left(|X_i|>a\right)$ is equal to $1$ if $|X_i|>a$, and is $0$ otherwise. Also, let $\hat{S}=\hat{X}_1+\hat{X}_2+\dots+\hat{X}_N$. We first note that the empirical average of the $\hat{X_i}$'s converge to its expectation. More formally, an application of the Bernstein concentration inequality (see e.g., \cite{concentration}) implies that for every $\epsilon'>0$ and for some $c_1>0$, the following holds:
\beq
\Prob\left(|\hat{S}/N - \stexp(\hat{S}/N)| > \epsilon' \right) < \exp(-c_1 N\epsilon').
\eeq
\noindent On the other hand:
\beq
\stexp(\hat{S}/N) = \stexp{\hat{X}_1} = \Prob(|X_1|>a) = \sqrt{2/\pi}e^{-\frac{a^2}{2}}.
\label{eq:SM'_conc}
\eeq

\noindent Similarly, for the random variable $S_N = X_1 + X_2 + \dots + X_N$, we can write the following concentration inequality using Chernoff bound for some $c_2 >0$:
\beq
\Prob\left(|S_N/N - \stexp(S_N/N)| > \epsilon' \right) < \exp(-c_2 N\epsilon').
\label{eq:SN_conc}
\eeq
\noindent Since $\stexp (S_N/N) = \sqrt{2/\pi}$, this establishes (\ref{eq:Order_stat1}).

Let the random variable $M'$ be the number of nonzero $\hat{X}_i$'s. First of all, note that $\hat{S} = S_{M'}$. The rest of the proof includes the following steps. We prove that $S_{M'}/S_{N}$  is concentrated around $\stexp{S_{M'}}/\stexp{S_N}$ with high probability. Then we use the fact that $M'$ also converges to its expected values, $M$, to show that $S_{M}/S_{N}$ becomes arbitrarily close to $S_{M'}/S_{N}$. As a result, $S_M/S_N$ will be concentrated around $\stexp{S_{M'}}/\stexp{S_N}$ with high probability, which is the desired result.

Concentration of $S_{M'}/S_N/$ is shown by using equations  (\ref{eq:SM'_conc}) and (\ref{eq:SN_conc}) simultaneously. Combining the two inequalities, we conclude that
\beq
\Prob\left(\left|\frac{S_{M'}}{N} - \sqrt{2/\pi}e^{-\frac{a^2}{2}}\right| \leq \epsilon'~\text{and}~\left|\frac{S_{N}}{N} - \sqrt{\frac{2}{\pi}}\right| \leq \epsilon'\right) \geq 1-e^{-c_1 N\epsilon'}-e^{-c_2 N\epsilon'},
\eeq
\noindent and thus,
\beq
\Prob\left( \frac{\sqrt{2/\pi}e^{-\frac{a^2}{2}}-\epsilon'}{\sqrt{2/\pi}+\epsilon'}\leq \frac{S_{M'}}{S_N} \leq \frac{\sqrt{2/\pi}e^{-\frac{a^2}{2}}+\epsilon'}{\sqrt{2/\pi}-\epsilon'}\right) \geq 1-e^{-c_1 N\epsilon'}-e^{-c_2 N\epsilon'},
\eeq
\noindent and consequently:
\beq
\Prob\left(  \left|\frac{S_{M'}}{S_N}- e^{-\frac{a^2}{2}}\right| \leq \frac{2\sqrt{2/\pi}(e^{-\frac{a^2}{2}}+1)\epsilon'}{2/\pi - \epsilon'^2}\right) \geq 1-e^{-c_1 N\epsilon'}-e^{-c_2 N\epsilon'}.
\eeq
\noindent If $\epsilon'$ is sufficiently small, then $ \frac{2\sqrt{2/\pi}(e^{-\frac{a^2}{2}}+1)\epsilon'}{2/\pi - \epsilon'^2} \leq \alpha\epsilon'$, for some constant $\alpha > 0$. Taking $\epsilon'' = \alpha\epsilon'$, $\alpha_1 = c_1/\alpha$ and $\alpha_2 = c_2/\alpha$, we can say that for sufficiently small $\epsilon''$ the following holds:
\beq
\Prob\left(  \left|\frac{S_{M'}}{S_N}- e^{-\frac{a^2}{2}}\right| \leq \epsilon''\right) \geq 1-e^{-\alpha_1 N\epsilon''}-e^{-\alpha_2 N\epsilon''}.
\label{eq:aux_concentr}
\eeq
\noindent Now we show that the quantity $\frac{|S_M-S_{M'}|}{S_N}$ will be arbitrarily small for large $N$. To do so, assume without loss of generality that $|X_1|\geq |X_2|\geq \dots \geq |X_N|$, and that $M_1 = \min(M,M')$, and $M_2 = \max(M,M')$. We then have:
\beq
|S_M-S_{M'}| = |X_{M_1 + 1}| + |X_{M_1 + 2}| + \dots + |X_{M_2}|,
\eeq
\noindent and
\bea
|S_N| &=& |X_1| + |X_2| + \dots + |X_N| \nonumber\\
 &\geq & |X_1| +|X_2| + \cdots |X_{M_1}| \nonumber\\
 &\geq & (N-M_{1})|X_{M_1}| \nonumber\\
&\geq& \frac{N-M_{1}}{M_2 - M_1}|S_{M_2}-S_{M_1}| \label{eq:SM,M'_aux}  \\
&\geq& \frac{N-M}{|M-M'|}|S_{M}-S_{M'}|. \label{eq:SM,M'_main}
\eea
\noindent Note that equation (\ref{eq:SM,M'_aux}) holds because $|X_{M_1}|$ is larger than all the values $|X_{M_1+1}|,\dots,|X_{M_2}|$, and is therefore larger than $1/(M_2-M_1)$ times their sum. It directly follows from (\ref{eq:SM,M'_main}) that:
\beq
\frac{|S_M-S_{M'}|}{S_N}\leq \frac{|M'-M|}{N-M}.\label{eq:SM,M'_main_main}
\eeq
Therefore, to show the concentration of the left hand side in the above inequality, it suffices to show that $\frac{|M'-M|}{N-M}$ concentrates. Since the variables $X_i'={\bf{1}}\left(|X_i|>a\right)$ are independent Bernoulli random variables with probability $2Q(a) = \frac{M}{N}$ of being nonzero, a Chernoff concentration bound on their empirical average implies that
\beq
\Prob(\left|\frac{\sum_{i=1}^n X_i'}{N} - \stexp X' \right| \leq \epsilon''' ) \geq 1-e^{-c_3\epsilon'''N},
\eeq
\noindent for some $c_3>0$, and for every $\epsilon'''>0$, where $X'$ has the same distribution as all $X_i'$'s. Noting that $\sum_{i=1}^nX_i' = M'$ and $\stexp X' = M/N$, the above implies that:
\beq
\Prob(\frac{|M-M'|}{N}\leq  \epsilon''')  = \Prob(\frac{|M-M'|}{N-M}\leq  \frac{1}{1-M/N}\epsilon''') \geq 1-e^{-c_3\epsilon'''N}.
\eeq

\noindent If the ratio $M/N$ is kept constant, the quantity $\frac{\epsilon'''}{1-M/N}$ will be smaller than any $\tilde{\epsilon} >0$ as $\epsilon'''$ becomes arbitrarily small, which shows the concentration of $\frac{|M-M'|}{N-M}$. Using this and
the inequality of (\ref{eq:SM,M'_main_main}) we can conclude that $\frac{|S_M-S_{M'}|}{S_N}\leq \tilde{\epsilon}$ with probability  $1-e^{-\alpha_3\tilde{\epsilon}N}$ for some constant $\alpha_3>0$. Combining this latter conclusion with (\ref{eq:aux_concentr}), it follows that
\beq
\Prob\left(  \left|\frac{S_{M}}{S_N}- e^{-\frac{a^2}{2}}\right| \leq \epsilon''+\epsilon'''\right) \geq 1-e^{-\alpha_1 N\epsilon''}-e^{-\alpha_2 N\epsilon''}-e^{-\alpha_3 N\tilde{\epsilon}}.
\eeq
\noindent Consequently, we conclude that if $\epsilon$ is sufficiently small, the following holds:
\beq
\Prob\left(  \left|\frac{S_{M}}{S_N}- e^{-\frac{a^2}{2}}\right| \leq \epsilon\right) \geq 1-3e^{-c N\epsilon},
\eeq
\noindent for some $c>0$, which concludes the proof of (\ref{eq:Order_stat2}).

\section{Proof of Lemma \ref{lem:W(x,a)}}
\label{app:proof_of_W(x,ax)}
Let $\beta = 1-2Q(\sqrt{-2\log(1-\alpha)})$, and without loss of generality assume that the $k$ nonzero values of $\x$ are $x_1,x_2,\dots,x_k$, with $|x_1|\leq |x_2|\leq \dots \leq x_k$. In order to show that $W(\x,\alpha\|\x\|_1) < k(\beta+\epsilon)$, it suffices to show that $\sum_{i=1}^{k(\beta + \epsilon)}|x_i| >\alpha \|\x\|_1$. Applying the order statistic result of Lemma \ref{lemma:Gaussian_Base}, we have that with high probability:
\beq
\frac{\sum_{i=1}^{k(\beta + \epsilon)}|x_i|}{\sum_{i=1}^{k}|x_i|} \approx 1-\exp(-\frac{\Psi(\frac{1-\beta-\epsilon}{2})}{2}) > 1-\exp(-\frac{\Psi(\frac{1-\beta}{2})}{2}) = f,
\eeq
which concludes the proof.
\section{Computation of $\lambda_c$ Threshold}
\label{App:delta_c}

In \cite{journalweighted}, a ``sectional'' threshold $\delta_c^{(T)}(\gamma_1,\gamma_2,f_1,f2_,\omega)$ is defined, with the following implication. Let $L$ be an index set of size $\gamma_1 n$. If $\delta \geq \delta_c^{(T)}(\gamma_1,\gamma_2,f_1,f2_,\omega)$, then a sparse vector $\x$ with a random sign pattern with
exactly $\gamma_1f_1 n$ nonzero entries over $L$ and exactly $\gamma_2f_2 n$ entries over $\overline{L}$ can be recovered using the following weighted $\ell_1$ minimization:
\beq \min\|\z_L\|_1+\omega\|\z_{\overline{L}}\|_1~~\text{subject to}~~ \A\z = \A\x \eeq
\noindent The reason $\delta_c^{(T)}$ is called sectional is that it provides a recovery guarantee for all support set $\x$ satisfying the nonuniform sparsity pattern, but almost all support sets. From this definition, it immediately follows that the $\lambda_c$ of Theorem \ref{thm:delta} is given by:

\beq \lambda_c = \max_{f'_1\geq f_1, f'_1k+f'_2(n-k) = k}\delta_c^{(T)}(\frac{k}{n},1-\frac{k}{n},f'_1,f'_2,\omega).
\label{delta'}
\eeq

\noindent Furthermore, The explicit derivation of $\delta_c^{T}$ is given in \cite{journalweighted} which is as follows:

\begin{align}
\nonumber \delta_c^{(T)} = &\min\{\delta~|~\psi_{com}(\tau_1,\tau_2)-\psi_{int}(\tau_1,\tau_2)-\psi_{ext}(\tau_1,\tau_2)<0~ \\
\nonumber &\forall ~0\leq \tau_1\leq \gamma_1(1-f_1), 0\leq \tau_2\leq \gamma_2(1-f_2),\\
\nonumber &\tau_1+\tau_2 > \delta-\gamma_1f_1-\gamma_2f_2 \}
\label{deltaT}
\end{align}
\noindent where $\psi_{com}$, $\psi_{int}$ and $\psi_{ext}$ are obtained as follows. Define $g(x)=\frac{2}{\sqrt{\pi}}e^{-{x^2}}$, $G(x)=\frac{2}{\sqrt{\pi}}\int_{0}^{x}e^{-y^2}dy$ and let $\varphi(.)$ and $\Phi(.)$ be the standard Gaussian pdf and cdf functions respectively.
\begin{align}
\nonumber &\psi_{com}(\tau_1,\tau_2) = (\tau_1+\tau_2 + \gamma_1(1-f_1)H(\frac{\tau_1}{\gamma_1(1-f_1)})\\
&+\gamma_2(1-f_2)H(\frac{\tau_2}{\gamma_2(1-f_2)}) + \gamma_1H(f_1) + \gamma_2H(f_2)) \log{2}
\end{align}
\noindent where $H(\cdot)$ is the Shannon entropy function. Define $c=(\tau_1+\gamma_1f_1)+\omega ^2(\tau_2+\gamma_2f_2)$, $\alpha_1=\gamma_1(1-f_1)-\tau_1$ and $\alpha_2=\gamma_2(1-f_2)-\tau_2$. Let $x_0$ be the unique solution to $x$ of the equation
$2c-\frac{g(x)\alpha_1}{xG(x)}-\frac{\omega g(\omega x)\alpha_2}{xG(\omega x)}=0$. Then
\begin{equation}
\psi_{ext}(\tau_1,\tau_2) = cx_0^2-\alpha_1\log{G(x_0)}-\alpha_2\log{G(\omega x_0)}
\end{equation}

\noindent Let $b=\frac{\tau_1+\omega ^2\tau_2}{\tau_1+\tau_2}$, $\Omega'=\gamma_1f_1+\omega ^2\gamma_2f_2$ and $Q(s)=\frac{\tau_1\varphi(s)}{(\tau_1+\tau_2)\Phi(s)}+\frac{\omega \tau_2\varphi(\omega s)}{(\tau_1+\tau_2)\Phi(\omega s)}$. Define the function $\hat{M}(s)=-\frac{s}{Q(s)}$ and solve for $s$ in $\hat{M}(s)=\frac{\tau_1+\tau_2}{(\tau_1+\tau_2)b+\Omega'}$. Let the unique solution be $s^*$ and set $y=s^*(b-\frac{1}{\hat{M}(s^*)})$. Compute the rate function $\Lambda^*(y)= sy -\frac{\tau_1}{\tau_1+\tau_2}\Lambda_1(s)-\frac{\tau_2}{\tau_1+\tau_2}\Lambda_1(\omega s)$ at the point $s=s^*$, where $\Lambda_1(s) = \frac{s^2}{2} +\log(2\Phi(s))$.
The internal angle exponent is then given by:
\begin{equation}
\psi_{int}(\tau_1,\tau_2) = (\Lambda^*(y)+\frac{\tau_1+\tau_2}{2\Omega'}y^2+\log2)(\tau_1+\tau_2)
\end{equation}

When $f_1\rightarrow 1$ and $f_2\rightarrow 0$, the terms $\lambda_c(\gamma_1,\gamma_2,f_1,f2_,\omega)$ and $\delta_c^{(T)}(\gamma_1,\gamma_2,f_1,f2_,\omega)$ become arbitrarily close, and converge to $\delta_c(\gamma_1,\gamma_2,1,0,\omega)$, which is defined as the weak threshold of weighted $\ell_1$ minimization for the weighted $\ell_1$ minimization for the nonuniform sparsity model with set fractions $\gamma_1,\gamma_2$ and sparsity fractions $1$ and $0$(see \cite{journalweighted}).

\section{Proof of Theorem \ref{thm:scale}}
\label{app:proof_of_scaling}

The proof of this theorem is common to the most part with the technical details of \cite{isitrobust}, which are based on Grassman manifold techniques for the performance analysis of compressed sensing. The method is basically the extension of the high dimensional techniques of Donoho \emph{et al.} \cite{D,D1} for incorporating noise into the performance bounds of $\ell_1$ minimization. First consider the following lemma.

\begin{lem}
Let $\A$ be a general $m\times n$ measurement matrix, $\x$ be an $n$-element
vector and $\y=\A\x$. Denote $K$ as a subset of $\{1,2,\dots,n\}$ such
that its cardinality $|K|=k$ and further denote $\overline{K}=\{1,2,\dots,n\}\setminus K$. Let $\w$
denote an $n \times 1$ vector. Let $C>1$ be a fixed number.

Given a specific set $K$ and suppose that the part of $\x$ on $K$, namely $\x_{K}$ is fixed.
$\forall \x_{\overline{K}}$, any solution $\hat{\x}$ produced by the $\ell_1$ minimization
satisfies
\begin{equation*}
 \|\x_K\|_1-\|\hat{\x}_K\|_{1} \leq
\frac{2}{C-1} \|\x_{\overline{K}}\|_1
\end{equation*}
 and
\begin{equation*}
 \|(\x-\hat{\x})_{\overline{K}}\|_1
\leq \frac{2C}{C-1} \|\x_{\overline{K}}\|_1,
\end{equation*}
if and only if $\forall \w\in \mathbb{R}^n~\mbox{such that}~\A\w=0$,
we have
 \begin{equation}
\|\x_K+\w_{K}\|_1+ \|\frac{\w_{\overline{K}}}{C}\|_1 \geq \|\x_K\|_1.
 \label{eq:Grasswthmeq1}
 \end{equation}
\end{lem}

In fact, if (\ref{eq:Grasswthmeq1}) is satisfied, we will have the stability result

\begin{equation*}
 \|(\x-\hat{\x})_{\overline{K}}\|_1
\leq \frac{2C}{C-1} \|\x_{\overline{K}}\|_1.
\end{equation*}

In \cite{isitrobust}, it was established that when the matrix $\A$ is sampled from an i.i.d. Gaussian ensemble, $C=1$, considering a single index set $K$, there exists a constant ratio $0<\mu_{W}<1$ such that if $\frac{|K|}{n} \leq \mu_{W}$, then with overwhelming probability as $n \rightarrow \infty$, the condition (\ref{eq:Grasswthmeq1}) holds for all $\w\in \mathbb{R}^n~\mbox{satisfying}~\A\w=0$. Now if we take a single index set $K$ with cardinality $\frac{|K|}{n}=(1-\epsilon_1){\mu_{W}}$, we would like to derive a characterization of $C$, as a function of $\frac{|K|}{n}=(1-\epsilon_1){\mu_{W}}$, such that the condition (\ref{eq:Grasswthmeq1}) holds for all $\w\in \mathbb{R}^n~\mbox{satisfying}~\A\w=0$.

When the measurement matrix $\A$ is sampled from an i.i.d. Gaussian ensemble, it is known that the probability that the condition (\ref{eq:Grasswthmeq1}) holds for all $\w \in \mathbb{R}^n~\mbox{satisfying}~\A\w=0$ is the \emph{Grassmann angle}, namely the probability that an $(n-m)$-dimensional uniformly distributed subspace intersects a polyhedral cone trivially (intersecting only at the apex of the cone). The complementary probability that the condition (\ref{eq:Grasswthmeq1}) does not hold for all $\w \in \mathbb{R}^n~\mbox{satisfying}~\A\w=0$ is the \emph{complementary Grassmann angle}. In our problem, without loss of generality, we scale $\x_{K}$ (extended to an $n$-dimensional vector supported on $K$) to a point in the relative interior of a $(k-1)$-dimensional face $F$ of the weighted $\ell_1$ ball,
\begin{equation}
\text{SP}=\{\y\in \mathbb{R}^n~|~\|\y_K\|_1+ \|\frac{\y_{\overline{K}}}{C}\|_1 \leq 1\}.
\end{equation}

The polyhedral cone we are interested in for the complementary Grassmann angle is the cone $\text{SP}-\x_{K}$,  namely the cone obtained by setting $\x_{K}$  as the apex, and observing $\text{SP}$ from this apex.

Building on the works by Santal\"{o} \cite{santalo} and
McMullen \cite{McMullen}  in high dimensional integral
geometry and convex polytopes, the complementary Grassmann angle for
the $(k-1)$-dimensional face $F$ can be explicitly expressed as the
sum of products of internal angles and external angles \cite{Grunbaumbook}:
\begin{equation}
P=2\times \sum_{s \geq 0}\sum_{G \in \Im_{m+1+2s}(\text{SP})}
{\beta(F,G)\gamma(G,\text{SP})}, \label{eq:Grassangformula}
\end{equation}
where $s$ is any nonnegative integer, $G$ is any
$(m+1+2s)$-dimensional face of the SP
($\Im_{m+1+2s}(\text{SP})$ is the set of all such faces),
$\beta(\cdot,\cdot)$ stands for the internal angle and
$\gamma(\cdot,\cdot)$ stands for the external angle.

The internal angles and external angles are basically defined as
follows \cite{Grunbaumbook}\cite{McMullen}:
\begin{itemize}
\item An internal angle $\beta(F_1, F_2)$ is the fraction of the
hypersphere $S$ covered by the cone obtained by observing the face
$F_2$ from the face $F_1$. \footnote{Note the dimension of the
hypersphere $S$ here matches the dimension of the corresponding cone
discussed. Also, the center of the hypersphere is the apex of the
corresponding cone. All these defaults also apply to the definition
of the external angles. } The internal angle $\beta(F_1, F_2)$ is
defined to be zero when $F_1 \nsubseteq F_2$ and is defined to be
one if $F_1=F_2$.
\item An external angle $\gamma(F_3, F_4)$ is the fraction of the
hypersphere $S$ covered by the cone of outward normals to the
hyperplanes supporting the face $F_4$ at the face $F_3$. The
external angle $\gamma(F_3, F_4)$ is defined to be zero when $F_3
\nsubseteq F_4$ and is defined to be one if $F_3=F_4$.

\end{itemize}

When $C=1$, we denote the probability $P$ in (\ref{eq:Grassangformula}) as $P_{1}$. By definition, the weak threshold $\mu_{W}$ is the supremum of $\frac{|K|}{n} \leq \mu_{W}$ such that the probability $P_{1}$ in (\ref{eq:Grassangformula}) goes to $0$ as $n \rightarrow \infty$. We need to show for $\frac{|K|}{n}=(1-\epsilon_1){\mu_{W}}$ and $C=\frac{1}{\sqrt{1-\epsilon_1}}$, (\ref{eq:Grassangformula}) also goes to $0$ as $n \rightarrow \infty$. To that end, we only need to show the probability $P'$ that, there exists an $\w$ from the null space of $A$ such that
 \begin{equation}
\|\x_K+\w_{K}\|_1+ \|\frac{\w_{\overline{K_1}}}{C_{\infty}}\|_1 +\|\frac{\w_{\overline{K_2}}}{C}\|_1 < \|\x_K\|_1
 \label{eq:Grasswthmeq2}
 \end{equation}
 goes to $0$ as $n \rightarrow \infty$, where $C_{\infty}$ is a large number which we may take as $\infty$ at the end, ${\overline{K_1}}$, ${\overline{K_2}}$ and $K$ are disjoint sets such that $|{\overline{K_1}}\bigcup {K}|=\mu_{W}n$ and ${\overline{K_1}}\bigcup {\overline{K_2}}=\overline{K}$.

Then the probability $P'$ will be equal to the probability that an $(n-m)$-dimensional uniformly distributed subspace intersects the polyhedral cone $\text{WSP}-\x_{K}$ nontrivially (intersecting at some other points besides the apex of the cone), where $\text{WSP}$ is the polytope

\begin{equation}
\text{WSP}=\{\y\in \mathbb{R}^n~|~\|\y_K\|_1+ \|\frac{\y_{\overline{K_1}}}{C_{\infty}}\|_1 +\|\frac{\y_{\overline{K_2}}}{C}\|_1 \leq 1\}.
\end{equation}

Then $P'$ is also a complementary Grassmann angle, which can be expressed by
\cite{Grunbaumbook}:
\begin{equation}
P'=2\times \sum_{s \geq 0}\sum_{G \in \Im_{m+1+2s}(\text{WSP})}
{\beta(F,G)\gamma(G,\text{WSP})}. \label{eq:Grassangformulanewp}
\end{equation}

Now we only need to show $P' \leq P_{1}$. If we denote $l=(m+1+2s)+1$ and $k=(1-\epsilon_1)\mu_{W}n$, in the polytope $\text{WSP}$, then there are in total $\binom{n-k}{l-k} 2^{l-k}$ faces $G$ of dimension
$(l-1)$ such that $F\subseteq G$ and $\beta(F, G) \neq 0$.

However, we argue that  when $C_{\infty}$ is very large, only $\binom{n-k_1}{l-k_1} 2^{l-k}$ such faces $G$ of dimension $(l-1)$ will contribute nonzero terms to $P'$ in (\ref{eq:Grassangformulanewp}), where $k_1=\mu_{W}n$. In fact, a certain $(l-1)$-dimensional face $G$ supported on the index set $L$ is the convex hull of $C_{i} e_{i}$, where $i \in L $, $C_{i}$ is the corresponding weighting for index $i$ (which is $1$ for the set $K$, $C_{\infty}$ for the set $\overline{K_1}$ and $C$ for the set $\overline{K_2}$ ), and $e_i$ is the standard unit coordinate vector. Now we show that if $\overline{K_1} \nsubseteq L$, the corresponding term in (\ref{eq:Grassangformulanewp}) for the face $G$ will be $0$ when $C_{\infty}$ is very large.

\begin{lem}
Suppose that $F$ is a $(k-1)$-dimensional face of \text{WSP} supported on the subset $K$ with $|K|=k$. Then the external angle
$\gamma(G, \text{WSP})$ between an $(l-1)$-dimensional face $G$ supported on the set $L$($F
\subseteq G$) and the  polytope $\text{WSP}$ is $0$ when  $\overline{K_1} \nsubseteq L$ and $C_{\infty}$ is large.
\label{lemma:external}
\end{lem}

\begin{proof}
Without loss of generality, assume $K=\{n-k+1, \cdots,n\}$. Consider
the $(l-1)$-dimensional face
\begin{equation*}
G=\text{conv}\{C_{n-l+1}\times e^{n-l+1}, ... ,C_{n-k}\times e^{n-k}, e^{n-k+1},
..., e^{n}\}
\end{equation*}
of $\text{WSP}$. The $2^{n-l}$ outward
normal vectors of the supporting hyperplanes of the facets
containing $G$ are given by
\begin{equation*}
\{\sum_{p=1}^{n-l} j_{p}e_p/{C_p}+\sum_{p=n-l+1}^{n-k} e_p/C_p+
\sum_{p=n-k+1}^{n} e_p, j_{p}\in\{-1,1\}\}.
\end{equation*}

Then the outward normal cone $c(G, \text{WSP})$ at the face $G$ is
the positive hull of these normal vectors. When $\overline{K_1} \nsubseteq L$, the fraction of the surface of the $(n-l-1)$-dimensional sphere taken by the cone $c(G, \text{WSP})$ is $0$ since the corresponding $C_{p}$ is very large.
\end{proof}

Now let us look at the internal angle $\beta(F,G)$ between the
$(k-1)$-dimensional face $F$ and an $(l-1)$-dimensional face $G$, where $\overline{K_1}$ is a subset of the support set of $G$.
Notice that the only interesting case is when $F \subseteq G$ since
$\beta(F,G)\neq 0$ only if $F \subseteq G$. We will see if $F
\subseteq G$, the cone $c(F,G)$ formed by observing $G$ from $F$ is the
direct sum of a $(k-1)$-dimensional linear subspace and the positive hull of $(l-k)$ vectors. These $(l-k)$ vectors are in the form
\begin{equation*}
v_{i}=(-\frac{1}{k},...,-\frac{1}{k},0,...,C_{i},0,...0), i \in L \setminus K.
\end{equation*}
 For those vectors $v_{i}$ with $i \in \overline{K_1}$, $C_{i}=C_{\infty}$. When $C_{\infty}$ is very large, the considered cone takes half of the space at each $i$-th coordinate with $i \in \overline{K_1}$.

So by the definition of the internal angle, the internal angle $\beta(F,G)$ is equal to $\frac{1}{2^{k_1-k}} \times \beta{(F,G_1)}$, where $G_1$ is supported only on the set $L\setminus \overline{K_1}$. It is known that this internal angle $\beta{(F,G_1)}$ is equal to the fraction of an $(l-k_1-1)$-dimensional sphere taken by a polyhedral cone formed by $(l-k_{1})$ unit vectors with inner product
$\frac{1}{1+C^2k}$ between each other. In this case, the internal
angle is given by
\begin{equation}
\beta(F,G)=\frac{1}{2^{k_1-k}} \frac{V_{l-k_1-1}(\frac{1}{1+C^2k},l-k_1-1)}{V_{l-k_1-1}(S^{l-k_1-1})},
\label{eq:internal}
\end{equation}
where $V_i(S^i)$ denotes the $i$-th dimensional surface measure on
the unit sphere $S^{i}$, while $V_{i}(\alpha', i)$ denotes the
surface measure for regular spherical simplex with $(i+1)$ vertices
on the unit sphere $S^{i}$ and with inner product as $\alpha'$
between these $(i+1)$ vertices. Thus (\ref{eq:internal}) is equal to
$B(\frac{1}{1+C^2k}, l-k_1)$, where
\begin{equation}
B(\alpha', m')=\theta^{\frac{m'-1}{2}} \sqrt{(m'-1)\alpha' +1}
\pi^{-m'/2} {\alpha'}^{-1/2}J(m',\theta),
\end{equation}
with $\theta=(1-\alpha')/\alpha'$ and
\begin{equation}
 J(m', \theta)=\frac{1}{\sqrt{\pi}} \int_{-\infty}^{\infty}(\int_{0}^{\infty} e^{-\theta v^2+2i v\lambda} \,dv )^{m'} e^{-\lambda^2} \,d\lambda.
\end{equation}

If we take $C=\frac{1}{\sqrt{1-\epsilon_1}}$, then
\begin{equation*}
\frac{1}{1+C^2k}=\frac{1}{1+k_1}.
\end{equation*}

By comparison, $\beta(F,G)=\frac{1}{2^{k_1-k}} \times \beta{(F,G)}$ is exactly the $\frac{1}{2^{k_1-k}} \beta(F_1, G_1)$ term appearing in the expression for the Grassmann angle $P$ between the face $F_1$ supported on the set $K_1$ and the polytope $\text{SP}$, where $G_1$ is an $(l-1)$-dimensional face of $\text{SP}$ supported on the set $L$.

Similar to the derivation for the internal angle, we can show that the external angle $\gamma(G, \text{WSP})$ is also exactly equal to $\gamma(G_1, \text{SP})$ term appearing in the expression for the Grassmann angle $P$ between the face $F_1$ supported on the set $K_1$ and the polytope $\text{SP}$, where $G_1$ an $(l-1)$-dimensional face of $\text{SP}$ supported on the set $L$.

Since there are in total only $\binom{n-k_1}{l-k_1} 2^{l-k}$ such faces $G$ of dimension $(l-1)$ will contribute nonzero terms to $P'$ in (\ref{eq:Grassangformulanewp}), substituting the results for the internal and external angles, we have $P=P'$. Thus for $\frac{|K|}{n}=(1-\epsilon_1) \mu_{W}$ and $C =\frac{1}{\sqrt{1-\epsilon_1}}$, with high probability, the condition the condition (\ref{eq:Grasswthmeq1}) holds for all $\w \in \mathbb{R}^n~\mbox{satisfying}~A\w=0$.

\end{document}